\documentclass{article}

\usepackage[T1]{fontenc}
\usepackage[sc]{mathpazo}
\usepackage{amsmath}
\usepackage{amssymb}
\usepackage{enumerate}
\usepackage{amsthm}
\usepackage{amsfonts,mathrsfs}
\usepackage[style]{fncychap}
\usepackage{graphicx, color} 
\usepackage{geometry} 
\usepackage{eepic}
\usepackage{ifthen}

\usepackage{subfigure}
\usepackage{multirow}
\newboolean{ElectronicVersion}
\setboolean{ElectronicVersion}{true} 

\usepackage[unicode=true,pdfusetitle, bookmarks=true,bookmarksnumbered=false,bookmarksopen=false, breaklinks=false,pdfborder={0 0 0},backref=false,colorlinks=false] {hyperref}
\hypersetup{
colorlinks,linkcolor=myurlcolor,citecolor=myurlcolor,urlcolor=myurlcolor}
\definecolor{myurlcolor}{rgb}{0,0,0.7}

\usepackage{hyperref}
\hypersetup{pdfpagemode=UseNone}

\newcommand{\tinyspace}{\mspace{1mu}}

\newcommand{\abs}[1]{\left\lvert\tinyspace #1 \tinyspace\right\rvert}

\renewcommand{\t}{{\scriptscriptstyle\mathsf{T}}}

\newcommand{\setft}[1]{\mathrm{#1}}

\newcommand{\unitary}[1]{\setft{U}\left(#1\right)}

\def\real{\mathbb{R}}
\def\natural{\mathbb{N}}

\def\I{\mathbb{1}}

\def \dif {\mathrm{d}}
\def \diag {\mathrm{diag}}

\def \Ad {\mathrm{Ad}}

\newenvironment{mylist}[1]{\begin{list}{}{
    \setlength{\leftmargin}{#1}
    \setlength{\rightmargin}{0mm}
    \setlength{\labelsep}{2mm}
    \setlength{\labelwidth}{8mm}
    \setlength{\itemsep}{0mm}}}
    {\end{list}}

\def\ot{\otimes}

\newcommand{\iinner}[2]{\langle #1 | #2\rangle}
\newcommand{\out}[2]{| #1\rangle\langle #2 |}

\newcommand{\Innerm}[3]{\left\langle #1 \left| #2 \right| #3 \right\rangle}

\newcommand{\pa}[1]{(#1)}
\newcommand{\Pa}[1]{\left(#1\right)}

\newcommand{\Br}[1]{\left[#1\right]}
\newcommand{\set}[1]{\{#1\}}
\newcommand{\Set}[1]{\left\{#1\right\}}

\newcommand{\ket}[1]{|#1\rangle}

\def\Jamiolkowski{J}
\newcommand{\jam}[1]{\Jamiolkowski\pa{#1}}

\DeclareMathOperator{\trace}{Tr}
\newcommand{\ptr}[2]{\trace_{#1}\pa{#2}}
\newcommand{\Ptr}[2]{\trace_{#1}\Pa{#2}}

\newcommand{\Tr}[1]{\Ptr{}{#1}}

\def\cE{\mathcal{E}}
\def\cF{\mathcal{F}}\def\cH{\mathcal{H}}

\def\cP{\mathcal{P}}
\def\cU{\mathcal{U}}

\def\bP{\mathbf{P}}\def\bQ{\mathbf{Q}}

\def\rF{\mathrm{F}}

\def\rS{\mathrm{S}}

\newcommand{\cmp}{Comm. Math. Phys.~}

\newcommand{\jmp}{J. Math. Phys.~}
\newcommand{\jpa}{J. Phys. A~}

\newcommand{\prl}{Phys. Rev. Lett.~}
\newcommand{\pra}{Phys. Rev. A~}
\newcommand{\pre}{Phys. Rev. E~}

\newcommand{\sci}{Science~}
\newcommand{\aop}{Ann. Phys.~}

\newtheorem{thrm}{Theorem}[section]
\newtheorem{lem}[thrm]{Lemma}
\newtheorem{prop}[thrm]{Proposition}

\theoremstyle{definition}

\newtheorem{remark}[thrm]{Remark}

\numberwithin{equation}{section}

\newcounter{questionnumber}

\begin{document}

\title{Average entropy of a subsystem over a global unitary orbit of a mixed bipartite state}

\author{Lin Zhang$^1$\footnote{E-mail: godyalin@163.com; linyz@zju.edu.cn}\,\, and\,  Hua Xiang$^2$\footnote{E-mail: hxiang@whu.edu.cn}\\
  {\it\small $^1$Institute of Mathematics, Hangzhou Dianzi University, Hangzhou 310018, PR~China}\\
  {\it\small $^2$School of Mathematics and Statistics, Wuhan University, Wuhan 430072, PR~China}}
\date{}
\maketitle
\maketitle \mbox{}\hrule\mbox\\
\begin{abstract}
We investigate the average entropy of a subsystem within a global
unitary orbit of a given mixed bipartite state in the
finite-dimensional space. Without working out the closed-form
expression of such average entropy for the mixed state case, we
provide an analytical lower bound for this average entropy. In
deriving this analytical lower bound, we get some useful by-products
of independent interest. We also apply these results to estimate
average correlation along a global unitary orbit of a given mixed
bipartite state. When the notion of von Neumann entropy is replaced
by linear entropy, the similar problem can be considered also, and
moreover the exact average linear entropy formula is derived for a
subsystem over a global unitary orbit of a mixed
bipartite state.
\end{abstract}
\maketitle \mbox{}\hrule\mbox

\section{Introduction}

In 1978, Lubkin \cite{Lub1978} proposed a method of approximating
the average entropy for a subsystem of a finite-dimensional quantum
system in a global pure bipartite state by expanding the entropy as
a series in terms of the average traces of powers of the system's
reduced density operator, but the convergence of this series was
never established. However, the author of recent paper
\cite{Dyer2014} found an exact closed form expression for the
average traces, in which Dyer gave a characterization of the
convergence of the series.

In fact, Page conjectured in \cite{Page1993} that if a quantum
system of Hilbert space dimension $mn$ is in a random pure bipartite
state, the average entropy of a subsystem of dimension $m\leqslant
n$ should be given by the simple and elegant formula
\begin{eqnarray}\label{eq:Pageformula}
S_{m,n} = H_{mn}- H_n - \frac{m-1}{2n},
\end{eqnarray}
where $H_k:=\sum^k_{j=1}\frac1j$ is the $k$-th harmonic number. The
average entropy $S_{m,n}$ in \eqref{eq:Pageformula} is also served
as a way of understanding the information in black hole radiation.
This formula was first proved by Foong and Kanno \cite{FK1994} by
using Fourier transform, and next by S\'{a}nchez-Ruiz \cite{SR1995}
and by Sen \cite{Sen1996} by using random matrix theory connected
with generalized Laguerre polynomials. Some years later, Lachal
\cite{AL2006} used a probabilistic approach to give a re-derivation
of Page's formula. Recently, Zhang \cite{lin2017} has shown that if
a quantum system of Hilbert space dimension $mn$ is in a random pure
bipartite state, the average diagonal entropy of a subsystem of
dimension $m\leqslant n$ should be given by the simple and elegant
formula
\begin{eqnarray}\label{eq:Linformula}
S^D_{m,n} = H_{mn}- H_n.
\end{eqnarray}
Based on the above mentioned formulas, i.e. \eqref{eq:Pageformula}
and \eqref{eq:Linformula}, he derives quickly that the average
coherence of a subsystem of dimension $m\leqslant n$ is given by
$S^D_{m,n}-S_{m,n}=\frac{m-1}{2n}$ (see also \cite{LUA2017} for
another approach).

We know that a random pure state can be generated by a unitary
operator chosen uniformly according to Haar measure $\mu$. Thus the
above problem can be equivalently described as follows: Consider a
complex quantum system $AB$ which consists of two subsystems $A$ and
$B$. For a given pure state $\rho_{AB}=\out{\psi_{AB}}{\psi_{AB}}$,
\begin{eqnarray*}
\int_{\unitary{d_Ad_B}} \rS(\ptr{B}{U\rho_{AB}U^\dagger})\dif\mu(U)=
H_{d_Ad_B} - H_{d_B} - \frac{d_A-1}{2d_B}\quad(d_A\leqslant d_B),
\end{eqnarray*}
where $\rS(\rho):=-\Tr{\rho\ln\rho}$ is the von Neumann entropy.
Along this line, in a very recent paper \cite{Ch2014}, Christandl
\emph{et al} computed exactly the eigenvalue distributions of
reduced density matrices of multipartite pure state by employing
symplectic geometric method. Here we ask: can one have an analogous
formula for a given mixed bipartite state $\rho_{AB}$? This question
has been paid  no or little attention to the best of our knowledge.
In this note, we will make an attempt to determine the average
entropy of a subsystem along a global unitary orbit of a given mixed
bipartite state.

In fact, recently, many researchers studied various problems along a
unitary orbit of a quantum state. For example, the total correlation
attained between the subsystems of a bipartite quantum system is
constrained if the bipartite system undergoes global unitary
evolutions. The authors of Ref~\cite{Jevtic2012a,Jevtic2012b}
investigated related problems motivated by some considerations in
the field of quantum thermodynamics. They have not only obtained the
value of the maximal quantum mutual information (QMI), but also the
maximum QMI state in the balanced bipartite quantum systems. Unlike
the maximum QMI case, finding the minimum QMI state on the unitary
orbit is more difficult than finding the maximum QMI state in
general. Luckily, they completely solved this minimum QMI state for
two-qubit case. Besides, Zhang and Fei investigated relative entropy
and fidelity between two unitary orbits of two states, respectively
\cite{ZF2014}. They have obtained a lot of compact expressions for
some extremal values under consideration. We have already known that
any bipartite quantum state can be diagonalized under the global
unitary conjugation but cannot be achieved in general under the
local unitary conjugations. Because of this, Zhang \emph{et al}
considered the fidelity between one bipartite quantum state and
another one undergoing local unitary dynamics \cite{ZCB2015}. The
problems are related to the geometric measure of entanglement and
the distillability problem. Besides the above mentioned works, there
is also one, where Oszmaniec and Ku\'{s} \cite{OK2014} estimated the
fraction of noncorrelated states within a unitary orbit of a given
mixed bipartite state. Based on the obtained result, they have
proven that the fraction of noncorrelated states tends to zero
exponentially fast with the dimension of the relevant Hilbert space
whenever the purity exceeding some critical value. Consequently, a
state within a global unitary orbit of a given bipartite state is
asymptotically a correlated one. Motivated by this and Page's
average entropy formula, in the present paper, we will consider the
calculation of average entropy along a global unitary orbit of a
given mixed bipartite state. On the one hand, our attempt made here
can be seen as a generalization of Page's formula. On the other
hand, the obtained result can be viewed as an estimate for the
entanglement within a global unitary orbit of a given mixed
bipartite state since a random state within a global unitary orbit
of a given bipartite state is asymptotically a correlated one. In
deriving our main results, we get a lot of by-products which is of
independent interest.

The paper is organized as follows. In Sect.~\ref{sect:ave-ent}, we
introduce the notion of unitary orbit for quantum states. In
calculating the average entropy of a subsystem along a global
unitary orbit of a given mixed bipartite state, we need to calculate
some integral (Lemma~\ref{lem:sixth-moment}), which is very
important in this paper, over unitary groups. Then, some
implications are discussed in Sect.~\ref{sect:implications}.
Specifically, we have obtained the following results: (i) For a
bipartite quantum system, maximally mixed state can be represented
by a uniform probability mixing of tensor products of two marginal
states at each point within the global unitary orbit of any given
mixed bipartite state; (ii) We use the relative entropy and fidelity
as a figure of merit for correlation, and estimate the average
correlation along a global unitary orbit from above and below since
analytical calculation seems unavailable; (iii) A detailed research
is given to the sum of average entropies of two subsystems within a
global unitary orbit. We summarize the main contents of this paper
in Sect.~\ref{sect:conclusion}. Finally, we present the detailed
proofs of Lemma~\ref{lem:sixth-moment}, Theorem~\ref{th:A1},
Proposition~\ref{prop:ave-corr}, and Theorem~\ref{th:A2} in
Appendix.

\section{Average entropy of a subsystem along a global unitary
orbit}\label{sect:ave-ent}

Define \emph{unitary orbit} $\cU_\rho$ of a given quantum state
$\rho$ on a $d$-dimensional Hilbert space $\cH_d$ as follows:
\begin{eqnarray*}
\cU_\rho:= \Set{U\rho U^\dagger: U\in\unitary{\cH_d}} \textcolor{blue}{.}
\end{eqnarray*}
Choose any $\rho'_{AB}\in \cU_{\rho_{AB}}$ for a given bipartite
state $\rho_{AB}$ and $\rho'_A = \ptr{B}{\rho'_{AB}}$. Let
$\Phi=\trace_B$, then $\Phi^*=\ot\I_B$. Assume that
$\Gamma:=\Phi^*\Phi$. That is $\Gamma(X)=\ptr{B}{X}\ot\I_B$.

In what follows, we compute the average entropy along the unitary
orbit $\cU_{\rho_{AB}}$ of $\rho_{AB}$. By the definition of
entropy, $\rS(\rho'_A) = -\Tr{\rho'_A\ln\rho'_A}$, which can be
rewritten as:
\begin{eqnarray*}
\rS(\rho'_A) = - \Tr{\rho_{AB}\ln \Br{U^\dagger
\Gamma(U\rho_{AB}U^\dagger)U}}.
\end{eqnarray*}
Note that $-\ln x = \sum^{\infty}_{n=1} \frac1n (1-x)^n$, and it
follows that
\begin{eqnarray*}
-\ln \Br{U^\dagger \Gamma(U\rho_{AB}U^\dagger)U} &=&
\sum^{\infty}_{n=1} \frac1n \Pa{\I_A\ot\I_B-\Br{U^\dagger
\Gamma(U\rho_{AB}U^\dagger)U}}^n\\
&=&\sum^{\infty}_{n=1} \frac1n U^\dagger\Br{\Gamma(UTU^\dagger)}^nU,
\end{eqnarray*}
where $T:=\I_A\ot\I_B/{d_B}-\rho_{AB}$ with $\Tr{T}=d_A-1$. Thus
\begin{eqnarray*}
\int_{\unitary{d}}\rS(\rho'_A)\dif\mu(U) = \sum^{\infty}_{n=1}
\frac1n \Tr{\rho \int
U^\dagger\Br{\Gamma(UTU^\dagger)}^nU\dif\mu(U)},
\end{eqnarray*}
where $d=d_Ad_B$. Denote
\begin{eqnarray*}
a_n := \Tr{\rho \int
U^\dagger\Br{\Gamma(UTU^\dagger)}^nU\dif\mu(U)}.
\end{eqnarray*}
Hence
\begin{eqnarray}
\ln(d_A)\geqslant\int\rS(\rho'_A)\dif\mu(U) = \sum^{\infty}_{n=1}
\frac{a_n}n.
\end{eqnarray}
We see from this formula
\begin{eqnarray}
\lim_{n\to\infty}a_n=0.
\end{eqnarray}
Theoretically, although the specific formula for $a_n$ can be
obtained \cite{Benoit,Collins,LZ2014}, however the specific form of
$a_n$ will be rather complicated. We can truncate this series (via
keeping the first $n$ terms) to obtain an approximation about the
average entropy. That is,
\begin{eqnarray}
\ln(d_A)\geqslant\int\rS(\rho'_A)\dif\mu(U) \geqslant a_1 +\frac12
a_2 + \cdots+\frac1na_n.
\end{eqnarray}
The more larger $n$ is, the more tighter lower bound is. But
however, the computing about $a_n$ is becoming very complicated when
$n\geqslant 3$. As an illustration, we truncate first two terms as
an estimate below. More challenging task is to determine the exact
value of the above series.\\~\\
(i). For $n=1$, we have \cite{LZ2014}
\begin{eqnarray*}
\int U^\dagger\Gamma(UTU^\dagger)U\dif\mu(U) =
\frac{d\Tr{\Gamma(\I_d)}-\Tr{\Gamma}}{d(d^2-1)}\Tr{T}\I_d +
\frac{d\Tr{\Gamma} - \Tr{\Gamma(\I_d)}}{d(d^2-1)}T.
\end{eqnarray*}
Since $\Gamma=\Phi^*\Phi$, choose any orthonormal basis
$\set{\ket{\varphi_j}: j=1,\dots,d_B}$ for subsystem space of $B$,
it follows that
$$
\Gamma(Z)=\sum^{d_B}_{i,j=1}M_{ij}ZM^\dagger_{ij},
$$
where $M_{ij}=\I_A\ot\out{\varphi_i}{\varphi_j}=M^\dagger_{ji}$,
which implies that $\Tr{\Gamma(\I_d)} =
d_Ad^2_B,\Tr{\Gamma}=d^2_Ad_B$. Thus
\begin{eqnarray}\label{eq:1st-matrix-int}
\int U^\dagger\Gamma(UTU^\dagger)U\dif\mu(U) =
\frac{d_A-1}{d^2-1}\Br{(1+dd_B)\I_A\ot\I_B - (d+d_B)\rho_{AB}}.
\end{eqnarray}
Therefore,
\begin{eqnarray}
a_1 = \frac{d_A-1}{d^2-1}\Br{(1+dd_B) - (d+d_B)\Tr{\rho^2_{AB}}}.
\end{eqnarray}
(ii). For $n=2$,
\begin{eqnarray*}
\int U^\dagger\Br{\Gamma(UTU^\dagger)}^2U\dif\mu(U)
&=&\sum^{d_B}_{i,j,k,l=1}\int U^\dagger M_{ij}UTU^\dagger
(M_{ji}M_{kl})UTU^\dagger M_{lk}U\dif\mu(U)\\
&=&\sum^{d_B}_{i,j,l=1}\int U^\dagger M_{ij}UTU^\dagger
M_{jl}UTU^\dagger M_{li}U\dif\mu(U).
\end{eqnarray*}
In order to calculate above integral, we will need the following
result, described as follows:
\begin{lem}\label{lem:sixth-moment}
It holds that
\begin{eqnarray}\label{eq:mainprop}
\int UAU^\dagger BUXU^\dagger CUDU^\dagger \dif\mu(U)=
\mu_1\cdot\I_d+\mu_2\cdot BC+\mu_3\cdot CB + \mu_4\cdot B +
\mu_5\cdot C,
\end{eqnarray}
where the coefficients $\mu_j(j=1,\ldots,5)$ can be found in the
Appendix A.
\end{lem}
\begin{proof}
See \textbf{Appendix A}.
\end{proof}
This leads to the following result:
\begin{thrm}\label{th:A1}
The average von Neumann entropy of a subsystem within the global
unitary orbit of a generic mixed bipartite state $\rho_{AB}$, a
density matrix on $\cH_d\equiv\cH_{d_A}\ot\cH_{d_B}$, is bounded
from below by a quantity:
\begin{eqnarray}
\int\rS(\rho'_A)\dif\mu(U) \geqslant
a_1+\cdots+\frac{a_n}n\quad(\forall n\in\natural),
\end{eqnarray}
where
\begin{eqnarray*}
a_n = \Tr{\rho \int
U^\dagger\Br{\Gamma(UTU^\dagger)}^nU\dif\mu(U)}\geqslant a^n_1.
\end{eqnarray*}
In particular, the first two terms can be given specifically:
\begin{eqnarray}
a_1 &=& \frac{d_A-1}{d^2-1}\Br{(1+dd_B) - (d+d_B)\Tr{\rho^2_{AB}}},\\
a_2 &=& \Pa{f + g\frac1{d_B} + h\frac1{d^2_B}} -
\Pa{g+h\frac2{d_B}}\Tr{\rho^2_{AB}} + h\Tr{\rho^3_{AB}},
\end{eqnarray}
where $f,g,h$ are given below:
\begin{eqnarray}
f&=& \frac{d_A(d^2_A-1)(d^2_B-1)}{(d^2-1)(d^2-4)}\Pa{d_A+d_B\Tr{\rho^2_{AB}}-2}+\frac{(d^2-2d^2_A-2)(d^2_B-1)}{(d^2-1)(d^2-4)}(d_A-1)^2,\\
g&=& = \frac{2d(d_A-1)(d^2_A-1)(d^2_B-1)}{(d^2-1)(d^2-4)},\\
h&=& = \frac{(d^2_A-1)(d^2_A-4)d^2_B}{(d^2-1)(d^2-4)}.
\end{eqnarray}
Moreover,
\begin{eqnarray}
\int\rS(\rho'_A)\dif\mu(U) \geqslant - \ln(1-a_1).
\end{eqnarray}
\end{thrm}

\begin{proof}
See \textbf{Appendix B}.
\end{proof}

\begin{remark}
In particular, for $d_A=d_B=2$, then $d=4$, we have
\begin{eqnarray}
a_1 = \frac35-\frac25\Tr{\rho^2_{AB}},\quad a_2=\frac3{10} -
\frac15\Tr{\rho^2_{AB}}.
\end{eqnarray}
This implies that
\begin{eqnarray}
\ln 2 \geqslant \int\rS(\rho'_A)\dif\mu(U) \geqslant
\max\Set{\frac9{10} - \frac35\Tr{\rho^2_{AB}}, \ln 5 - \ln
\Pa{2+2\Tr{\rho^2_{AB}}}}.
\end{eqnarray}
It is left \emph{open} that the explicit computing of the following
integral for a given mixed bipartite state $\rho_{AB}$:
\begin{eqnarray}
\int \ln \Pa{U^\dagger \Gamma(U\rho_{AB}U^\dagger)U}\dif\mu(U).
\end{eqnarray}
\end{remark}

\begin{remark}
We have already known that for any super-operator $\Xi$ over $\cH_d$
\cite{LZ2014},
\begin{eqnarray}
\int_{\unitary{d}}\dif\mu(U) U^\dagger\Xi(UXU^\dagger)U =
\frac{d\Tr{\Xi(\I_d)}-\Tr{\Xi}}{d(d^2-1)}\Tr{X}\I_d +
\frac{d\Tr{\Xi}-\Tr{\Xi(\I_d)}}{d(d^2-1)}X.
\end{eqnarray}
Now let $d=d_Ad_B$ and $\cH_d=\cH_A\ot\cH_B$ with $\dim(\cH_A)=d_A$
and $\dim(\cH_B)=d_B$. Assume that $X = \rho_{AB}$, a density matrix
on $\cH_A\ot\cH_B$. Fixing an orthonormal basis
$\set{\ket{\psi_{B,j}}:j=1,\ldots,d_B}$ for $\cH_B$. Suppose that
$\Gamma(X) = \ptr{B}{X}\ot\I_B$. Then it can be rewritten as:
$$
\Gamma(X) =
\sum^{d_B}_{i,j=1}(\I_A\ot\out{\psi_{B,i}}{\psi_{B,j}})X(\I_A\ot\out{\psi_{B,j}}{\psi_{B,i}}).
$$
Clearly $\Gamma(\I_A\ot\I_B) = d_B\I_A\ot\I_B$, implying that
$$
\Tr{\Gamma(\I_A\ot\I_B)} = d_Ad^2_B~~\text{and}~~\Tr{\Gamma} =
\sum^{d_B}_{i,j=1}(d_A\delta_{ij})^2= d^2_Ad_B.
$$
From the above discussion, we see that
\begin{eqnarray}
\int_{\unitary{d}}\dif\mu(U) U^\dagger\Gamma(U\rho_{AB}U^\dagger)U =
\frac{dd_B-d_A}{d^2-1}\I_A\ot\I_B + \frac{dd_A-d_B}{d^2-1}\rho_{AB}.
\end{eqnarray}
Denote $\rho'_{AB} = U\rho_{AB}U^\dagger$ and
$\rho'_{A}=\ptr{B}{\rho'_{AB}}$. Then
\begin{eqnarray*}
\Tr{(\rho'_A)^2} &=& \Tr{(\rho'_A\ot\I_B)\rho'_{AB}} =
\Tr{\Gamma(\rho'_{AB})\rho'_{AB}}\\
&=&\Tr{U^\dagger\Gamma(U\rho_{AB}U^\dagger)U\rho_{AB}}.
\end{eqnarray*}
Therefore
\begin{eqnarray*}
\left\langle\Tr{(\rho'_A)^2}\right\rangle :=
\int\dif\mu(U)\Tr{(\rho'_A)^2} = \Tr{\int\dif\mu(U)
U^\dagger\Gamma(U\rho_{AB}U^\dagger)U\rho_{AB}}.
\end{eqnarray*}
That is,
\begin{eqnarray}\label{eq:two-purity}
\left\langle\Tr{(\rho'_A)^2}\right\rangle = \frac{dd_B-d_A}{d^2-1}+
\frac{dd_A-d_B}{d^2-1}\Tr{\rho^2_{AB}}.
\end{eqnarray}
If we use the notion of linear entropy $\rS_L(\rho):=1-\Tr{\rho^2}$,
then we have seen from Eq.~\eqref{eq:two-purity} that the average
\emph{linear entropy} is given precisely by
\begin{eqnarray}
\int\rS_L(\rho'_A)\dif\mu(U) = \frac{(d_A-1)(d_B-1)}{d+1} +
\frac{dd_A-d_B}{d^2-1}\rS_L(\rho_{AB}).
\end{eqnarray}
We see from this that
\begin{eqnarray}
\int\rS_L(\rho'_B)\dif\mu(U) = \frac{(d_A-1)(d_B-1)}{d+1} +
\frac{dd_B-d_A}{d^2-1}\rS_L(\rho_{AB}).
\end{eqnarray}
Moreover for a mixed state $\rho_{AB}$,
$\int\rS_L(\rho'_A)\dif\mu(U)=\int\rS_L(\rho'_B)\dif\mu(U)$ if and
only if $d_A=d_B$. We also see that
\begin{eqnarray}
\int(\rS_L(\rho'_A)+\rS_L(\rho'_B))\dif\mu(U) =
2\frac{(d_A-1)(d_B-1)}{d+1} + \frac{d_A+d_B}{d+1}\rS_L(\rho_{AB}).
\end{eqnarray}
In particular, if $d_A=d_B=2$, then
\begin{eqnarray}
\int\rS_L(\rho'_A)\dif\mu(U) = \int\rS_L(\rho'_B)\dif\mu(U) =
\frac15 + \frac25\rS_L(\rho_{AB}).
\end{eqnarray}
\end{remark}

\section{Some implications}\label{sect:implications}

Some consequences of Lemma~\ref{lem:sixth-moment} can be presented
as follows: Let $\Phi=\trace_B$ and $\Psi=\trace_A$. Then
$\Phi^*=\ot \I_B$ and $\Psi^*=\I_A\ot$. Thus for
$\rho'_{AB}=U\rho_{AB}U^\dagger$,
\begin{eqnarray}
&&\rho'_A\ot \rho'_B =
\Phi^*\Phi(U\rho_{AB}U^\dagger)\Psi^*\Psi(U\rho_{AB}U^\dagger)\\
&&=\sum_{i,j=1}^{d_B}\sum_{\mu,\nu=1}^{d_A}
(\I_A\ot\out{i}{j})U\rho_{AB}U^\dagger(\out{\mu}{\nu}\ot\out{j}{i})U\rho_{AB}U^\dagger(\out{\nu}{\mu}\ot\I_B),
\end{eqnarray}
and it follows that
\begin{eqnarray*}
\int (\rho'_A\ot \rho'_B) \dif\mu(U)
=\sum_{i,j=1}^{d_B}\sum_{\mu,\nu=1}^{d_A}
(\I_A\ot\out{i}{j})\Pa{\int
U\rho_{AB}U^\dagger(\out{\mu}{\nu}\ot\out{j}{i})U\rho_{AB}U^\dagger
dU}(\out{\nu}{\mu}\ot\I_B),
\end{eqnarray*}
leading to the following identity
\begin{eqnarray}
\int (\rho'_A\ot \rho'_B) \dif\mu(U) =
\frac{\I_A}{d_A}\ot\frac{\I_B}{d_B}.
\end{eqnarray}
This amounts to say that
\begin{prop}
For a bipartite quantum system, maximally mixed state can be
represented by a uniform probability mixing of tensor products of
two marginal states at each point within the global unitary orbit of
any mixed bipartite state.
\end{prop}
Besides, we also get that
\begin{eqnarray*}
&&\int U^\dagger(\rho'_A\ot \rho'_B)U \dif\mu(U) \\
&&= \sum_{i,j=1}^{d_B}\sum_{\mu,\nu=1}^{d_A} \int
U^\dagger(\I_A\ot\out{i}{j})
U\rho_{AB}U^\dagger(\out{\mu}{\nu}\ot\out{j}{i})U\rho_{AB}U^\dagger
(\out{\nu}{\mu}\ot\I_B)U\dif\mu(U),
\end{eqnarray*}
leading to the following identity
\begin{eqnarray}\label{eq:new-id}
\int U^\dagger(\rho'_A\ot \rho'_B)U \dif\mu(U) = c_0\cdot\I_d +
c_1\cdot\rho_{AB} + c_2\cdot\rho^2_{AB},
\end{eqnarray}
where
\begin{eqnarray}
c_0 &=&
\frac{(d^2_A-1)(d^2_B-1)}{(d^2-1)(d^2-4)}\Pa{d-2\Tr{\rho^2_{AB}}},\label{eq:c0}\\
c_1 &=& \frac{d^2(d^2_A+d^2_B-6)+4}{(d^2-1)(d^2-4)},\label{eq:c1}\\
c_2 &=& \frac{2d(d^2_A-1)(d^2_B-1)}{(d^2-1)(d^2-4)}.\label{eq:c2}
\end{eqnarray}

The total correlation attained between the subsystems of a bipartite
quantum system is constrained if the bipartite system undergoes
global unitary evolutions. The authors of
Ref~\cite{Jevtic2012a,Jevtic2012b} investigated related problems
from the field of quantum thermodynamics. Since knowing the maximal
possible variation in correlations is useful, it raises the
optimization problem, where a search of the maximal and minimal
correlated states on a unitary orbit is needed. This is completely
solved for two-qubit systems.

In the following, we make an attempt to calculate the average
quantum mutual information (QMI) along a unitary orbit of a generic
mixed bipartite state. We provide some analytical upper and/or lower
bounds for the average QMI and quantum fidelity along a unitary
orbit, although we cannot obtain precise formulas for them.

\begin{prop}\label{prop:ave-corr}
A lower bound for the maximal correlation (defined by relative
entropy) within the global unitary orbit of a mixed bipartite state
$\rho_{AB}$ is given as follows
\begin{eqnarray}\label{eq:averageQMI-lower}
\int I(A:B)_{\rho'}\dif\mu(U)\geqslant\rS\Pa{\rho_{AB}
||c_0\cdot\I_d + c_1\cdot\rho_{AB} + c_2\cdot\rho^2_{AB}},
\end{eqnarray}
where $I(A:B)_\rho:=\rS(\rho_{AB}||\rho_A\ot\rho_B)$, where
$\rS(\rho||\sigma):=\Tr{\rho(\ln\rho-\ln\sigma)}$ is the relative
entropy. Similarly, the lower and upper bounds for the average
fidelity within the global unitary orbit of a mixed bipartite state
$\rho_{AB}$ are given as follows
\begin{eqnarray}\label{eq:FidelityLowerUpperBounds}
c_0+c_1\Tr{\rho^2_{AB}}+c_2\Tr{\rho^3_{AB}}\leqslant\int\rF(\rho'_{AB},\rho'_A\ot\rho'_B)\dif\mu(U)\leqslant
\rF(\rho_{AB},c_0\cdot\I_d + c_1\cdot\rho_{AB} +
c_2\cdot\rho^2_{AB}).
\end{eqnarray}
where $\rF(\rho,\sigma):=\Tr{\sqrt{\sqrt{\rho}\sigma\sqrt{\rho}}}$
is the fidelity between two states $\rho$ and $\sigma$.
\end{prop}

\begin{proof}
See \textbf{Appendix C}.
\end{proof}

\begin{remark}
Recently, in \cite{Giorda2017} the authors introduced a general
measure of correlations for two-qubit states based on the classical
mutual information between local observables. They focus on
(classical) correlations between sets of local observables instead
of the quantum vs classical distinction. In this perspective,
quantum states can be characterized as a whole by the average amount
of (classical) correlations between all pairs of local observables.
Under some restrictions, the authors calculated the average mutual
information, whose value depends on the state purity, and the
symmetry of the correlations distribution.
\end{remark}

For the two-qubit case, we apply the spectral decomposition
$\rho_{AB}=U\Lambda U^\dagger$ with  $\Lambda =
\diag(\lambda_1,\lambda_2,\lambda_3,\lambda_4)$, where
\begin{eqnarray*}
(\lambda_1,\lambda_2,\lambda_3,\lambda_4)^\t\in\Delta_3 =\Set{(p_1,p_2,p_3,p_4)^\t\in\real^4: p_j\geqslant0 (\forall
j=1,2,3,4),\sum^4_{j=1}p_j=1}.
\end{eqnarray*}
Due to the invariance under unitary conjugation, we have
\begin{eqnarray*}
\rF( \Lambda, c_0 + c_1 \Lambda + c_2 \Lambda^2) = \Tr{ \sqrt{
\sqrt{\Lambda} (c_0 + c_1 \Lambda + c_2 \Lambda^2) \sqrt{\Lambda} }
} = \sum_{j=1}^4 (c_0 \lambda_j + c_1 \lambda_j^2 + c_2
\lambda_j^3)^\frac{1}{2}.
\end{eqnarray*}

The lower bound in \eqref{eq:FidelityLowerUpperBounds} becomes
\begin{eqnarray*}
c_0+c_1\Tr{\Lambda^2}+c_2\Tr{\Lambda^3} = c_0 + \sum_{j=1}^4 ( c_1
\lambda_j^2 + c_2 \lambda_j^3 ).
\end{eqnarray*}

We numerically check the bounds in
\eqref{eq:FidelityLowerUpperBounds} by choosing 10000 random points
in the probability simplex $\Delta_3$ and drawing the 3D scatter
plot of  $( \cP(\Lambda), c_0+c_1 \Tr{\Lambda^2}+c_2\Tr{\Lambda^3},
\rF(\Lambda, c_0+c_1\Lambda+c_2\Lambda^2) )$. For clarity, its
three-view drawings are also given by using the corresponding 2D
scatter plots: $\cP(\Lambda)$ vs. $\rF(\Lambda,
c_0+c_1\Lambda+c_2\Lambda^2)$ in Figure \ref{Fig_FidelityBound}(b);
$c_0+c_1\Tr{\Lambda^2}+c_2\Tr{\Lambda^3}$ vs. $\rF(\Lambda,
c_0+c_1\Lambda+c_2\Lambda^2)$ in Figure \ref{Fig_FidelityBound}(c);
and $\cP(\Lambda)$ vs.  $c_0+c_1\Tr{\Lambda^2}+c_2\Tr{\Lambda^3}$ in
Figure \ref{Fig_FidelityBound}(d).
We also notice that the maximum of $c_0+c_1
\Tr{\Lambda^2}+c_2\Tr{\Lambda^3}$ is even less than the minimum of
$\rF(\Lambda, c_0+c_1\Lambda+c_2\Lambda^2)$.

\begin{figure}[htbp] \centering
\includegraphics[width=\textwidth]{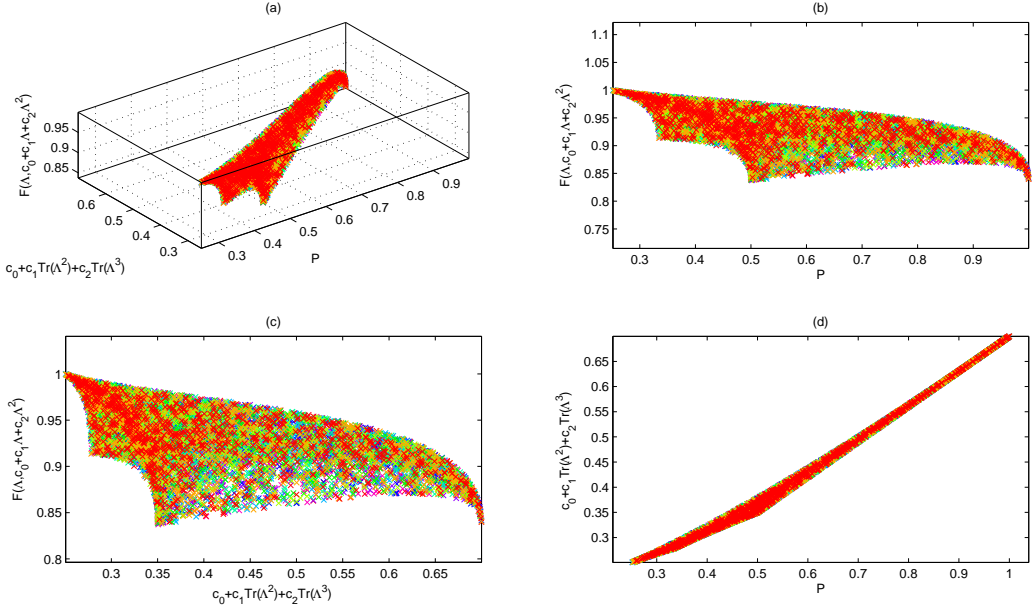}
\caption{ Two-qubit case. (a) 3D scatter plot of ( $\cP(\Lambda), c_0+c_1
\Tr{\Lambda^2}+c_2\Tr{\Lambda^3}, \rF(\Lambda,
c_0+c_1\Lambda+c_2\Lambda^2)$ ). 2D scatter plots: (b)
$\cP(\Lambda)$ vs. $\rF(\Lambda, c_0+c_1\Lambda+c_2\Lambda^2)$; (c)
$c_0+c_1\Tr{\Lambda^2}+c_2\Tr{\Lambda^3}$ vs. $\rF(\Lambda,
c_0+c_1\Lambda+c_2\Lambda^2)$; (d) $\cP(\Lambda)$ vs.
$c_0+c_1\Tr{\Lambda^2}+c_2\Tr{\Lambda^3}$.  }
\label{Fig_FidelityBound}
\end{figure}

In fact, not only the value of the maximal QMI is obtained, but also
the maximum QMI state $\rho_{\max}$ is derived in
\cite{Jevtic2012a,Jevtic2012b}. Specifically, $\rho_{AB}=U\Lambda
U^\dagger$ for $d_A=d_B$, where $\Lambda=\set{\lambda_j:
j=1,\ldots,d=d_Ad_B}$ is the spectrum of $\rho_{AB}$,
\begin{eqnarray}
\max_{\rho'_{AB}\in\cU_{\rho_{AB}}} I(A:B)_{\rho'} =
I(A:B)_{\rho_{\max}} = \ln (d)-\rS(\Lambda),
\end{eqnarray}
where
\begin{eqnarray}
\rho_{\max} = \sum^{d}_{j=1}\lambda_j\out{\Omega_j}{\Omega_j}
\end{eqnarray}
for any generalized Bell-state basis $\Set{\ket{\Omega_j}}$. Unlike
the maximum QMI case, finding $\rho_{\min}$ on the unitary orbit is
more difficult than finding $\rho_{\max}$ in general. Since QMI
varies on a unitary orbit completely depends on the sum of two
marginal entries, we need to figure out the spectra of two reduced
states of all the states in the unitary orbit $\cU_\rho$. The
literature indicates that calculating this set of compatible reduced
states with a given spectrum of a global bipartite seems
unforeseeable. This is well-known "quantum marginal problem", which
is fully solved for two-qubit case \cite{Bravyi2004} and
theoretically for two-qubit case \cite{Klyachko2006}. Specifically,
for two-qubit case, the solution of quantum marginal problem is
given by the following: Mixed two-qubit state $\rho_{AB}$ with
spectrum
$\Lambda=\set{\lambda_1\geqslant\lambda_2\geqslant\lambda_3\geqslant\lambda_4\geqslant0}$
and the marginal states $\rho_{A}$ and $\rho_B$ exist if and only if
minimal eigenvalues $\lambda^A_{\min}$ and $\lambda^B_{\min}$ of the
marginal states satisfy the inequalities \cite{Bravyi2004}:
\begin{eqnarray}
\begin{cases} \min\Pa{\lambda^A_{\min},\lambda^B_{\min}}\geqslant \lambda_3+\lambda_4,\\
\lambda^A_{\min}+\lambda^B_{\min}\geqslant
\lambda_2+\lambda_3+2\lambda_4,\\
\abs{\lambda^A_{\min}-\lambda^B_{\min}}\leqslant
\min\Pa{\lambda_1-\lambda_3,\lambda_2-\lambda_4}.
\end{cases}
\end{eqnarray}
With the help of this result, the value of the minimal QMI on a
unitary orbit is derived for two-qubit case
\cite{Jevtic2012a,Jevtic2012b}:
\begin{eqnarray}
I(A:B)_{\rho_{\min}} =
h(\lambda_1+\lambda_2)+h(\lambda_1+\lambda_3)-\rS(\Lambda),
\end{eqnarray}
where
$\lambda_1\geqslant\lambda_2\geqslant\lambda_3\geqslant\lambda_4\geqslant0$
and $h(p):=p\ln p+(1-p)\ln(1-p)$ is the binary entropy function
defined for $p\in[0,1]$, moreover two-qubit state $\rho_{\min}$ is
given by
\begin{eqnarray}
\rho_{\min} = \sum^0_{i,j=1}\lambda_{ij}\out{ij}{ij}.
\end{eqnarray}
Here $\lambda_{ij}$ is a re-indexing of $\lambda_k
(\lambda_{00}=\lambda_1,\lambda_{01}=\lambda_2,\ldots)$ and
$\set{\ket{i}},\set{\ket{j}}$ are qubit basis states for subsystems
$A$ and $B$.

By Eq.~\eqref{eq:averageQMI-lower} and the concavity of von Neumann
entropy, we can draw the following conclusion:
\begin{thrm}\label{th:A2}
For any generic mixed bipartite state $\rho_{AB}$, denoting
$$
\langle S_A+S_B\rangle:=\int\rS(\rho'_A)\dif\mu(U) +
\int\rS(\rho'_B)\dif\mu(U),
$$
we have that
\begin{eqnarray}
\rS(\rho_{AB})+\rS\Pa{\rho_{AB} ||c_0\cdot\I_d + c_1\cdot\rho_{AB} +
c_2\cdot\rho^2_{AB}}\leqslant \langle S_A+S_B\rangle \leqslant
\rS\Pa{c_0\cdot\I_d + c_1\cdot\rho_{AB} + c_2\cdot\rho^2_{AB}},
\end{eqnarray}
where $c_j(j=0,1,2)$ are given by \eqref{eq:c0}, \eqref{eq:c1}, and
\eqref{eq:c2}, respectively. Furthermore,
\begin{eqnarray*}
\rS\Pa{\rho_{AB} ||c_0\cdot\I_d + c_1\cdot\rho_{AB} +
c_2\cdot\rho^2_{AB}}\leqslant \int I(A:B)_{\rho'}\dif\mu(U)\leqslant
\rS\Pa{c_0\cdot\I_d + c_1\cdot\rho_{AB} + c_2\cdot\rho^2_{AB}} -
\rS(\rho_{AB}).
\end{eqnarray*}
In particular, we get a universal entropy inequality:
\begin{eqnarray}
\rS\Pa{\rho_{AB} ||c_0\cdot\I_d + c_1\cdot\rho_{AB} +
c_2\cdot\rho^2_{AB}}\leqslant \rS\Pa{c_0\cdot\I_d +
c_1\cdot\rho_{AB} + c_2\cdot\rho^2_{AB}} - \rS(\rho_{AB}).
\end{eqnarray}
Moreover, $\rho_{AB} = c_0\cdot\I_d + c_1\cdot\rho_{AB} +
c_2\cdot\rho^2_{AB}$ if and only if $\rho_{AB}$ is maximally mixed
state, i.e., $\rho_{AB}=\frac{\I_A}{d_A}\ot\frac{\I_B}{d_B}$.
\end{thrm}

\begin{proof}
See \textbf{Appendix D}.
\end{proof}
This result tells us whenever the state $\rho_{AB}$ is not maximally
mixed state, we can always find correlated states on a global
unitary orbit of this state. Indeed,  all quantum states are
unitarily connected to classical states, one-way or fully classical
\cite{Modi}.

In particular, for the balanced bipartite system, i.e., $d_A=d_B$,
we have, for $d=d_Ad_B$,
\begin{eqnarray*}
\int\rS(\rho'_A)\dif\mu(U) + \int\rS(\rho'_B)\dif\mu(U) &\leqslant&
\max_{U}\Pa{\rS(\rho'_A)+\rS(\rho'_B)} =
\max_{U}I(A:B)_{\rho'}+\rS(\rho_{AB}) = \ln(d),
\end{eqnarray*}
where we used the fact that
$\max_{U}I(A:B)_{\rho'}=\ln(d)-\rS(\rho_{AB})$ \cite{Jevtic2012a}.
Apparently, this upper bound is true for $d_A\neq d_B$, and moreover
it is trivially since $\rS(\rho'_X)\leqslant \ln (d_X)$, where
$X=A,B$. Thus the following inequality is always true
\begin{eqnarray}
\int\rS(\rho'_A)\dif\mu(U) + \int\rS(\rho'_B)\dif\mu(U)
\leqslant\ln(d).
\end{eqnarray}
In view of this, we get a tighter upper bound for the sum of the
average entropy of two subsystems:
\begin{eqnarray}
\int\rS(\rho'_A)\dif\mu(U) + \int\rS(\rho'_B)\dif\mu(U) \leqslant
\rS\Pa{c_0\cdot\I_d + c_1\cdot\rho_{AB} +
c_2\cdot\rho^2_{AB}}\leqslant\ln(d).
\end{eqnarray}

We will check the difference between the upper bound
$\rS\Pa{c_0\cdot\I_d + c_1\cdot\rho_{AB} + c_2\cdot\rho^2_{AB}}$ and
the lower bound $\rS(\rho_{AB})+ \rS\Pa{\rho_{AB} ||c_0\cdot\I_d +
c_1\cdot\rho_{AB} + c_2\cdot\rho^2_{AB}} $ since we want to identify
the range of the sum of two average entropy. Let
\begin{eqnarray}
\cF(\rho_{AB}):= \rS\Pa{c_0\cdot\I_d + c_1\cdot\rho_{AB} +
c_2\cdot\rho^2_{AB}} - \rS(\rho_{AB}) -\rS\Pa{\rho_{AB}
||c_0\cdot\I_d + c_1\cdot\rho_{AB} + c_2\cdot\rho^2_{AB}}
\end{eqnarray}
and $\cP=\Tr{\rho^2_{AB}}$, which lies in $\Br{d^{-1},1}$ for
$d=d_Ad_B$. Apparently, $\cF\geqslant0$ over the whole set of
states.

If $d_A=d_B=2$, then
\begin{eqnarray}
c_0 =\frac{2-\cP}{10},\quad c_1 = \frac15,\quad c_2 = \frac25.
\end{eqnarray}
Thus in the two-qubit case, we have
\begin{eqnarray}
\int U^\dagger(\rho'_A\ot \rho'_B)U \dif\mu(U) =
\frac{2-\cP}{10}\cdot\I_4 + \frac15\cdot\rho_{AB} +
\frac25\cdot\rho^2_{AB}.
\end{eqnarray}
Note that $\rS(\rho_{AB})+ \rS\Pa{\rho_{AB} ||c_0\cdot\I_d +
c_1\cdot\rho_{AB} + c_2\cdot\rho^2_{AB}} = - \Tr{\rho_{AB}
\ln\Br{c_0\cdot\I_d + c_1\cdot\rho_{AB} + c_2\cdot\rho^2_{AB} }}$.
Now $\cF$ is reduced to the following form:
\begin{eqnarray}
\cF(\rho_{AB})&=&\rS\Pa{\frac{2-\cP}{10}\cdot\I_4 +
\frac15\cdot\rho_{AB} + \frac25\cdot\rho^2_{AB}}  \notag\\
&&+ \Tr{\rho_{AB} \ln\Br{\frac{2-\cP}{10}\cdot\I_d +
\frac15\cdot\rho_{AB} + \frac25\cdot\rho^2_{AB}}}.
\end{eqnarray}
Since $\cF$ is invariant under unitary conjugation, it follows, via
the spectral decomposition $\rho=U\Lambda U^\dagger$ of $\rho_{AB}$,
where $\Lambda = \diag(\lambda_1,\lambda_2,\lambda_3,\lambda_4)$,
that
\begin{eqnarray}
\cF(\Lambda)&=& \rS\Pa{\frac{2-\cP}{10}\cdot\I_4 +
\frac15\cdot\Lambda + \frac25\cdot\Lambda^2}  \notag\\
&&+\Tr{\Lambda\ln\Br{\frac{2-\cP}{10}\cdot\I_4 + \frac15\cdot\Lambda
+ \frac25\cdot\Lambda^2}}.
\end{eqnarray}
We see that $\cF$ is a symmetric function defined over the
probability simplex
$$
\Delta_3=\Set{(p_1,p_2,p_3,p_4)^\t\in\real^4: p_j\geqslant0 (\forall
j=1,2,3,4),\sum^4_{j=1}p_j=1}.
$$
Next we make numerical analysis of this function defined over the
probability simplex.
\begin{eqnarray}
\cF(\lambda_1,\lambda_2,\lambda_3,\lambda_4)&=& \sum^4_{j=1}
\lambda_j\ln\Pa{\frac{2-\cP}{10}  + \frac15\lambda_j+\frac25\lambda^2_j} \\
&&- \Pa{\frac{2-\cP}{10} +
\frac15\lambda_j+\frac25\lambda^2_j}\ln\Pa{\frac{2-\cP}{10} +
\frac15\lambda_j+\frac25\lambda^2_j},
\end{eqnarray}
where $(\lambda_1,\lambda_2,\lambda_3,\lambda_4)^\t\in\Delta_3$ and
$\cP\equiv\cP(\Lambda):=\Tr{\Lambda^2}=\sum^4_{j=1}\lambda^2_j$.

Choose random point $\Lambda$ in the probability simplex $\Delta_3$,
and then draw the 3-dimensional figure about the 3-tuples $(x,y,z)$
where $x=\cP(\Lambda)$, $y=
-\Tr{\Lambda\ln\Br{\frac{2-\cP(\Lambda)}{10}\cdot\I_4 +
\frac15\cdot\Lambda + \frac25\cdot\Lambda^2}}$,
$z=\rS(c_0+c_1\Lambda+c_2\Lambda^2)$.
To illustrate the relationship among $\cP(\Lambda)$,
$\rS(\rho_{AB})+ \rS\Pa{\rho_{AB} ||c_0\cdot\I_d + c_1\cdot\rho_{AB}
+ c_2\cdot\rho^2_{AB}}$ and $\rS(c_0+c_1\Lambda+c_2\Lambda^2)$, we
show the 3D scatter plot in Figure \ref{Fig_EntropyBoundScatter}(a),
where there are 5000 random points used in total. To make things
more clear, we also demonstrate 2D scatter plots: $\rS(\rho_{AB})+
\rS\Pa{\rho_{AB} ||c_0\cdot\I_d + c_1\cdot\rho_{AB} +
c_2\cdot\rho^2_{AB}} $ vs. $\rS(c_0+c_1\Lambda+c_2\Lambda^2)$ in
Figure \ref{Fig_EntropyBoundScatter}(b), $\cP$ vs.
$\rS(c_0+c_1\Lambda+c_2\Lambda^2)$ in Figure
\ref{Fig_EntropyBoundScatter}(c), and $\cP$ vs. $\cF$ in Figure
\ref{Fig_EntropyBoundScatter}(d).
Note that for the case where $\lambda_j=\frac{1}{4} (j=1,2,3,4)$ we have
$\rS(\Lambda)=\rS(c_0+c_1\Lambda+c_2\Lambda^2)=\ln 4 = 1.39$ and   $\cF=0$; while  the case with one $\lambda_j$ equaling to 1 and the other three being zeros yields $\rS(\Lambda)=0$, $\rS(c_0+c_1\Lambda+c_2\Lambda^2)=
-\frac{3}{10} \ln \frac{1}{10} - \frac{7}{10} \ln \frac{7}{10}=0.94$, $\Tr{\Lambda\ln(c_0+c_1\Lambda+c_2\Lambda^2) } = \ln \frac{7}{10}$ and  $\cF=-\frac{3}{10} \ln \frac{1}{10} - \frac{3}{5} \ln \frac{7}{10}=0.58$.
In these plots, $\cP \in [\frac{1}{4}, 1]$,
$\rS(\Lambda) \in [0, \ln 4]$,  $\rS(c_0+c_1\Lambda+c_2\Lambda^2)
\in [0.94, \ln 4]$ and $\cF \in [0,0.58]$.

\begin{figure}[htbp] \centering
\includegraphics[width=\textwidth]{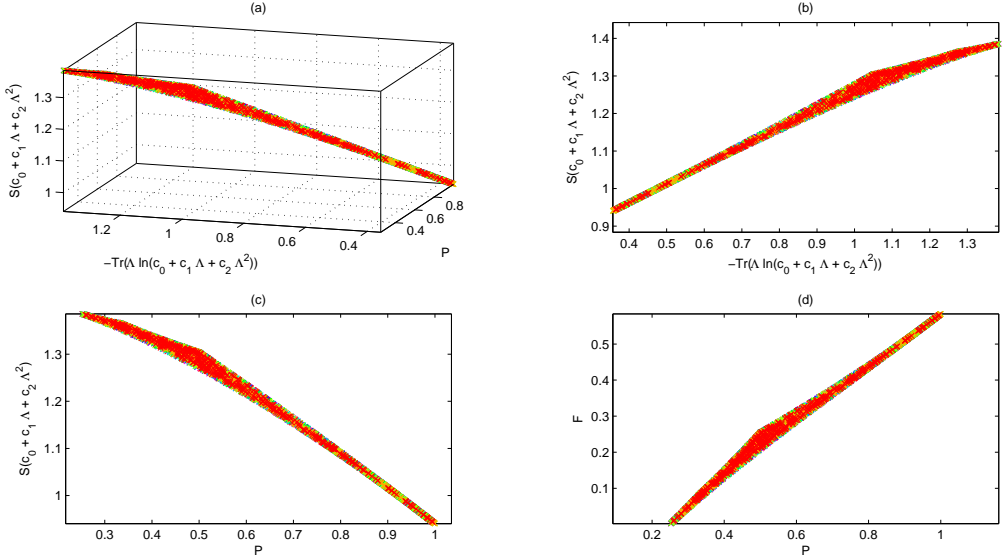}
\caption{ Two-qubit case ($d_A=d_B=2$). (a) 3D scatter plot of ( $\cP(\Lambda), -\Tr{\Lambda\ln(c_0+c_1\Lambda+c_2\Lambda^2) },
\rS(c_0+c_1\Lambda+c_2\Lambda^2)$ ). 2D scatter plots: (b) $-\Tr{\Lambda\ln(c_0+c_1\Lambda+c_2\Lambda^2) }$ vs. $\rS(c_0+c_1\Lambda+c_2\Lambda^2)$; (c)  $\cP(\Lambda)$ vs. $\rS(c_0+c_1\Lambda+c_2\Lambda^2)$; (d) $\cP(\Lambda)$ vs.  $\cF(\Lambda)$.  } \label{Fig_EntropyBoundScatter}
\end{figure}

\begin{remark}
Denote $\rho'=U\rho U^\dagger$ for a fixed state $\rho$. If $U$ is
such that $I(A:B)_{\rho'}>I(A:B)_{\rho}$, then we say that $U$
enhances the correlation between $A$ and $B$, otherwise, we say that
it weakens the correlation $A$ and $B$ when
$I(A:B)_{\rho'}<I(A:B)_{\rho}$. In particular, if
$I(A:B)_{\rho'}=0$, then we say that $U$ decouples $A$ from $B$.
Clearly there exists a state (for example completely mixed state) of
which its correlations  cannot be enhanced by any unitary. We see
from previous discussion, finding minimal mutual information is in
terms of an optimization problem that is extremely hard to handle in
higher dimensions. Proposition~\ref{prop:ave-corr} and
Theorem~\ref{th:A2} can be seen as a different strategy based on
probabilistic averages. Denote
$$
\Delta(\rho_{AB}) := \int I(A:B)_{\rho'}\dif\mu(U) - I(A:B)_{\rho}.
$$
If $\Delta(\rho)>0$, then the correlations existing in $\rho_{AB}$
can be enhanced by many unitaries (via a perspective of
concentration of measure phenomenon from Measure Theory). Clearly a
lot of states of product forms satisfies $\Delta(\rho)>0$. In the
opposite, $\Delta(\rho)<0$ means that the correlations existing in
$\rho_{AB}$ can be decreased by many unitaries. Naturally, some
questions arise: For a given state $\rho_{AB}$, one wants to know
wether if its unitary orbit of $\rho_{AB}$ contains product
state(s).
\end{remark}

\begin{remark}
A similar questions can be considered for the so-called
\emph{quantum conditional mutual information} (QCMI), defined by
$I(A:B|E)_\rho:= \rS(\rho_{AE})+\rS(\rho_{BE}) -
\rS(\rho_{ABE})-\rS(\rho_E)$ for a tripartite $\rho_{ABE}$ on a
tensorial Hilbert space $\cH_A\ot\cH_B\ot\cH_E$. We still denote
$\rho'_{ABE}=(U\ot\I_E)\rho_{ABE} (U^\dagger\ot\I_E)$ for any
unitary $U$ on $\cH_A\ot\cH_B$. Analogously, maximal and/or minimal
QCMI can be considered, that is,
$$
\max_{U}I(A:B|E)_{\rho'}\quad\text{and}\quad\min_{U}I(A:B|E)_{\rho'}.
$$
Clearly $\min_{U}I(A:B|E)_{\rho'}$ is very important since it gives
a lower bound for QCMI:
$$
I(A:B|E)_\rho\geqslant \min_{U}I(A:B|E)_{\rho'}.
$$
We can also consider the average QCMI:
\begin{eqnarray}
\int I(A:B|E)_{\rho'} \dif\mu(U).
\end{eqnarray}
All problems mentioned above are beyond the goal of this paper. We
will come back to them in the future research.
\end{remark}

\begin{remark}
Given a quantum channel $\cE$ (trace-preserving and completely
positive linear map) which is, via Kraus representation, represented
as $\cE=\sum_{j}\Ad_{E_j}$, where $\Ad_{E_j}(X):=E_jXE^\dagger_j$.
We can use Lemma~\ref{lem:sixth-moment} to get the average purity of
a unitary orbit of a given state undergoing a fixed quantum channel
$\cE$. Indeed, by Choi-Jamio{\l}kowski isomorphism \cite{Choi1975},
$\jam{\cE}:=(\cE\ot\I)(\out{\I_d}{\I_d})=\sum_j \out{E_j}{E_j}$,
\begin{eqnarray}
\int \cE(U\rho U^\dagger)^2 \dif\mu(U) &=& \sum_{i,j} E_i\Pa{U\rho
U^\dagger
E^\dagger_iE_jU\rho U^\dagger}E^\dagger_j\\
&=&\frac{d\Tr{\rho^2}-1}{d(d^2-1)}\sum_{i,j}\Tr{E^\dagger_iE_j}E_iE^\dagger_j
+ \frac{d-\Tr{\rho^2}}{d(d^2-1)}\cE(\I_d)^2,
\end{eqnarray}
implying that
\begin{eqnarray}
\int \Tr{\cE(U\rho U^\dagger)^2} \dif\mu(U)
=\frac{d\Tr{\rho^2}-1}{d(d^2-1)}\Tr{\jam{\cE}^2} +
\frac{d-\Tr{\rho^2}}{d(d^2-1)}\Tr{\cE(\I_d)^2},
\end{eqnarray}
where the lhs of the above last formula is just the average purity
of a unitary orbit of a given state undergoing a fixed quantum
channel $\cE$.
\end{remark}

\section{Concluding remarks}\label{sect:conclusion}

In this paper, we investigate the average entropy of a subsystem
along a global unitary orbit of a given mixed bipartite state in the
finite-dimensional space. Although it is still unable to derive the
closed-form compact formula for the mixed state case, compared with
Page's formula, nevertheless we get an analytical lower bound for
this average entropy for the mixed state case. In deriving this
analytical lower bound, we obtain some useful by-products of
independent interest, for instance, for a bipartite quantum system,
maximally mixed state can be represented by a uniform probability
mixing of tensor products of two marginal states at each point
within the global unitary orbit of any given mixed bipartite state.
This is amazing. In addition, from the discussion after finishing
the proof of Proposition~\ref{prop:ave-corr}, we see that the
average entropy of a subsystem is intimately related to the
well-known "quantum marginal problem" or $N$-representability in
quantum chemistry. Besides, it also connects with entanglement
polytope, a notion proposed by M. Walter \emph{et al.} in studying
multipartite entanglement from single-particle information
\cite{Walter2013}.

These obtained results can be applied to estimate average
correlation along a global unitary orbit of a given mixed bipartite
state. The corresponding numerics about these results is also
provided in lower dimensional case (for instance $d_A=d_B=2$).
Except that, the results obtained in the present paper can also be
used to study further the average coherence of a class of random
states induced from isospectral bipartite mixed states. Indeed,
recently we have already calculated exactly the average coherence
for random mixed quantum states \cite{lin2017,LUA2017} induced from
random bipartite pure states.

Finally, we conclude this section with two open problems: (i)
computing the average coherence (via relative entropy of coherence)
of a subsystem of isospectral bipartite systems
\begin{eqnarray}
\int\Pa{\rS((\rho'_A)_\diag) - \rS(\rho'_A)}\dif\mu(U),
\end{eqnarray}
and (ii) proving the following identity for the balanced system,
i.e., $d_A=d_B$,
\begin{eqnarray}\label{eq:conjecture2017}
\int\rS(\rho'_A)\dif\mu(U)=\int\rS(\rho'_B)\dif\mu(U).
\end{eqnarray}
If \eqref{eq:conjecture2017} \emph{were} true, then we would see
from Theorem~\ref{th:A1} and Theorem~\ref{th:A2} that
\begin{eqnarray}
-\ln(1-a_1)\leqslant \int\rS(\rho'_A)\dif\mu(U)\leqslant
\frac12\rS\Pa{c_0\cdot\I_d + c_1\cdot\rho_{AB} +
c_2\cdot\rho^2_{AB}}.
\end{eqnarray}
We hope that the present work and questions proposed can bring out
more interesting and insightful perspective(s) in quantum
information theory.

\subsubsection*{Acknowledgements}
L. Zhang is supported Natural Science Foundation of Zhejiang
Province of China (LY17A010027) and also by National Natural Science
Foundation of China (Nos.11301124 \& 61673145). H. Xiang is
supported by the National Natural Science Foundation of China
(Nos.11571265 \& 11471253). Michael Walter is also acknowledged for
his comments on this manuscript.


\section*{Appendix}\label{sect:appendix}

\subsection*{A. The proof of
Lemma~\ref{lem:sixth-moment}}

In this paper, we will utilize some notion of matrix integral
\cite{Benoit,Collins,LZ2014}. The formula in
Lemma~\ref{lem:sixth-moment} is given firstly. A detailed reasoning
is presented here.

\begin{proof}[The proof of Lemma~\ref{lem:sixth-moment}]
Firstly we note that
\begin{eqnarray*}
&&\Innerm{i_1}{UAU^\dagger BUXU^\dagger CUDU^\dagger}{i'_1} =
\sum_{j_1,j'_1} \Innerm{i_1}{U}{j_1}\Innerm{j_1}{AU^\dagger
BUXU^\dagger CUD}{j'_1}\Innerm{j'_1}{U^\dagger}{i'_1} \\
&&=\sum_{i_2,j_1,j_2,i'_2,j'_1,j'_2}
U_{i_1j_1}\overline{U}_{i'_1j'_1}\Innerm{j_1}{A}{j'_2}\Innerm{j'_2}{U^\dagger}{i'_2}
\Innerm{i'_2}{BUXU^\dagger C}{i_2}\Innerm{i_2}{U}{j_2}\Innerm{j_2}{D}{j'_1}\\
&&=\sum_{i_2,j_1,j_2,i'_2,j'_1,j'_2}
U_{i_1j_1}U_{i_2j_2}\overline{U}_{i'_1j'_1}\overline{U}_{i'_2j'_2}\Innerm{j_1}{A}{j'_2}
\Innerm{i'_2}{BUXU^\dagger C}{i_2}\Innerm{j_2}{D}{j'_1}\\
&&=\sum_{i_2,i_3,j_1,j_2,j_3,i'_2,i'_3,j'_1,j'_2,j'_3}
U_{i_1j_1}U_{i_2j_2}\overline{U}_{i'_1j'_1}\overline{U}_{i'_2j'_2}\\
&&~~~\times\Innerm{j_1}{A}{j'_2}
\Innerm{i'_2}{B}{i_3}\Innerm{i_3}{U}{j_3}\Innerm{j_3}{X}{j'_3}\Innerm{j'_3}{U^\dagger}{i'_3}\Innerm{i'_3} {C}{i_2}\Innerm{j_2}{D}{j'_1}\\
&&=\sum_{i_2,i_3,j_1,j_2,j_3,i'_2,i'_3,j'_1,j'_2,j'_3}
U_{i_1j_1}U_{i_2j_2}U_{i_3j_3}\overline{U}_{i'_1j'_1}\overline{U}_{i'_2j'_2}\overline{U}_{i'_3j'_3}\Innerm{j_1}{A}{j'_2}
\Innerm{i'_2}{B}{i_3}\Innerm{j_3}{X}{j'_3}\Innerm{i'_3}
{C}{i_2}\Innerm{j_2}{D}{j'_1}.
\end{eqnarray*}
Then we have:
\begin{eqnarray*}
&&\Innerm{i_1}{\int UAU^\dagger BUXU^\dagger CUDU^\dagger dU}{i'_1}\\
&& =\sum_{i_2,i_3,j_1,j_2,j_3,i'_2,i'_3,j'_1,j'_2,j'_3}A_{j_1,j'_2}
B_{i'_2,i_3}X_{j_3,j'_3}C_{i'_3,i_2}D_{j_2,j'_1}
\Pa{\int U_{i_1j_1}U_{i_2j_2}U_{i_3j_3}\overline{U}_{i'_1j'_1}\overline{U}_{i'_2j'_2}\overline{U}_{i'_3j'_3}dU}\\
&&= \sum_{i_2,i_3,j_1,j_2,j_3,i'_2,i'_3,j'_1,j'_2,j'_3} A_{j_1,j'_2}
B_{i'_2,i_3}X_{j_3,j'_3}C_{i'_3,i_2}D_{j_2,j'_1} \\
&&~~~\times\Pa{\sum_{\pi,\sigma\in S_3}
\iinner{i_1}{i'_{\pi(1)}}\iinner{i_2}{i'_{\pi(2)}}\iinner{i_3}{i'_{\pi(3)}}
\iinner{j_1}{j'_{\sigma(1)}}\iinner{j_2}{j'_{\sigma(2)}}\iinner{j_3}{j'_{\sigma(3)}}\mathrm{Wg}(\sigma\pi^{-1})}\\
&& = \sum_{\pi,\sigma\in S_3} \mathrm{Wg}(\sigma\pi^{-1})\\
&&~~~\times\Pa{\sum_{i_2,i_3,j_1,j_2,j_3,i'_2,i'_3,j'_1,j'_2,j'_3}
A_{j_1,j'_2}B_{i'_2,i_3}X_{j_3,j'_3}C_{i'_3,i_2}D_{j_2,j'_1}\iinner{i_1}{i'_{\pi(1)}}\iinner{i_2}{i'_{\pi(2)}}\iinner{i_3}{i'_{\pi(3)}}
\iinner{j_1}{j'_{\sigma(1)}}\iinner{j_2}{j'_{\sigma(2)}}\iinner{j_3}{j'_{\sigma(3)}}}.
\end{eqnarray*}
In what follows, we compute this value step-by-step. \\
(1). If $(\pi,\sigma)=((1),(1))$, then
\begin{eqnarray}
&&\sum_{i_2,i_3,j_1,j_2,j_3,i'_2,i'_3,j'_1,j'_2,j'_3}
A_{j_1,j'_2}B_{i'_2,i_3}X_{j_3,j'_3}C_{i'_3,i_2}D_{j_2,j'_1}
\iinner{i_1}{i'_1}\iinner{i_2}{i'_2}\iinner{i_3}{i'_3}
\iinner{j_1}{j'_1}\iinner{j_2}{j'_2}\iinner{j_3}{j'_3}\\
&&= \Tr{AD}\Tr{X}\Tr{BC}\iinner{i_1}{i'_{1}}.
\end{eqnarray}
(2). If $(\pi,\sigma)=((1),(12))$, then
\begin{eqnarray}
&&\sum_{i_2,i_3,j_1,j_2,j_3,i'_2,i'_3,j'_1,j'_2,j'_3}
A_{j_1,j'_2}B_{i'_2,i_3}X_{j_3,j'_3}C_{i'_3,i_2}D_{j_2,j'_1}
\iinner{i_1}{i'_1}\iinner{i_2}{i'_2}\iinner{i_3}{i'_3}
\iinner{j_1}{j'_2}\iinner{j_2}{j'_1}\iinner{j_3}{j'_3}\\
&&= \Tr{A}\Tr{D}\Tr{X}\Tr{BC}\iinner{i_1}{i'_{1}}.
\end{eqnarray}
(3). If $(\pi,\sigma)=((1),(13))$, then
\begin{eqnarray}
&&\sum_{i_2,i_3,j_1,j_2,j_3,i'_2,i'_3,j'_1,j'_2,j'_3}
A_{j_1,j'_2}B_{i'_2,i_3}X_{j_3,j'_3}C_{i'_3,i_2}D_{j_2,j'_1}
\iinner{i_1}{i'_1}\iinner{i_2}{i'_2}\iinner{i_3}{i'_3}
\iinner{j_1}{j'_3}\iinner{j_2}{j'_2}\iinner{j_3}{j'_1}\\
&&= \Tr{ADX}\Tr{BC}\iinner{i_1}{i'_1}.
\end{eqnarray}
(4). If $(\pi,\sigma)=((1),(23))$, then
\begin{eqnarray}
&&\sum_{i_2,i_3,j_1,j_2,j_3,i'_2,i'_3,j'_1,j'_2,j'_3}
A_{j_1,j'_2}B_{i'_2,i_3}X_{j_3,j'_3}C_{i'_3,i_2}D_{j_2,j'_1}
\iinner{i_1}{i'_1}\iinner{i_2}{i'_2}\iinner{i_3}{i'_3}
\iinner{j_1}{j'_1}\iinner{j_2}{j'_3}\iinner{j_3}{j'_2}\\
&&= \Tr{DAX}\Tr{BC}\iinner{i_1}{i'_1}.
\end{eqnarray}
(5). If $(\pi,\sigma)=((1),(123))$, then
\begin{eqnarray}
&&\sum_{i_2,i_3,j_1,j_2,j_3,i'_2,i'_3,j'_1,j'_2,j'_3}
A_{j_1,j'_2}B_{i'_2,i_3}X_{j_3,j'_3}C_{i'_3,i_2}D_{j_2,j'_1}
\iinner{i_1}{i'_1}\iinner{i_2}{i'_2}\iinner{i_3}{i'_3}
\iinner{j_1}{j'_2}\iinner{j_2}{j'_3}\iinner{j_3}{j'_1}\\
&&= \Tr{A}\Tr{DX}\Tr{BC}\iinner{i_1}{i'_1}.
\end{eqnarray}
(6). If $(\pi,\sigma)=((1),(132))$, then
\begin{eqnarray}
&&\sum_{i_2,i_3,j_1,j_2,j_3,i'_2,i'_3,j'_1,j'_2,j'_3}
A_{j_1,j'_2}B_{i'_2,i_3}X_{j_3,j'_3}C_{i'_3,i_2}D_{j_2,j'_1}
\iinner{i_1}{i'_1}\iinner{i_2}{i'_2}\iinner{i_3}{i'_3}
\iinner{j_1}{j'_3}\iinner{j_2}{j'_1}\iinner{j_3}{j'_2}\\
&&= \Tr{D}\Tr{AX}\Tr{BC}\iinner{i_1}{i'_1}.
\end{eqnarray}
(7). If $(\pi,\sigma)=((12),(1))$, then
\begin{eqnarray}
&&\sum_{i_2,i_3,j_1,j_2,j_3,i'_2,i'_3,j'_1,j'_2,j'_3}
A_{j_1,j'_2}B_{i'_2,i_3}X_{j_3,j'_3}C_{i'_3,i_2}D_{j_2,j'_1}
\iinner{i_1}{i'_2}\iinner{i_2}{i'_1}\iinner{i_3}{i'_3}
\iinner{j_1}{j'_1}\iinner{j_2}{j'_2}\iinner{j_3}{j'_3}\\
&&= \Tr{AD}\Tr{X}\Innerm{i_1}{BC}{i'_1}.
\end{eqnarray}
(8). If $(\pi,\sigma)=((12),(12))$, then
\begin{eqnarray}
&&\sum_{i_2,i_3,j_1,j_2,j_3,i'_2,i'_3,j'_1,j'_2,j'_3}
A_{j_1,j'_2}B_{i'_2,i_3}X_{j_3,j'_3}C_{i'_3,i_2}D_{j_2,j'_1}
\iinner{i_1}{i'_2}\iinner{i_2}{i'_1}\iinner{i_3}{i'_3}
\iinner{j_1}{j'_2}\iinner{j_2}{j'_1}\iinner{j_3}{j'_3}\\
&&= \Tr{A}\Tr{D}\Tr{X}\Innerm{i_1}{BC}{i'_1}.
\end{eqnarray}
(9). If $(\pi,\sigma)=((12),(13))$, then
\begin{eqnarray}
&&\sum_{i_2,i_3,j_1,j_2,j_3,i'_2,i'_3,j'_1,j'_2,j'_3}
A_{j_1,j'_2}B_{i'_2,i_3}X_{j_3,j'_3}C_{i'_3,i_2}D_{j_2,j'_1}
\iinner{i_1}{i'_2}\iinner{i_2}{i'_1}\iinner{i_3}{i'_3}
\iinner{j_1}{j'_3}\iinner{j_2}{j'_2}\iinner{j_3}{j'_1}\\
&&= \Tr{ADX}\Innerm{i_1}{BC}{i'_1}.
\end{eqnarray}
(10). If $(\pi,\sigma)=((12),(23))$, then
\begin{eqnarray}
&&\sum_{i_2,i_3,j_1,j_2,j_3,i'_2,i'_3,j'_1,j'_2,j'_3}
A_{j_1,j'_2}B_{i'_2,i_3}X_{j_3,j'_3}C_{i'_3,i_2}D_{j_2,j'_1}
\iinner{i_1}{i'_2}\iinner{i_2}{i'_1}\iinner{i_3}{i'_3}
\iinner{j_1}{j'_1}\iinner{j_2}{j'_3}\iinner{j_3}{j'_2}\\
&&= \Tr{DAX}\Innerm{i_1}{BC}{i'_1}.
\end{eqnarray}
(11). If $(\pi,\sigma)=((12),(123))$, then
\begin{eqnarray}
&&\sum_{i_2,i_3,j_1,j_2,j_3,i'_2,i'_3,j'_1,j'_2,j'_3}
A_{j_1,j'_2}B_{i'_2,i_3}X_{j_3,j'_3}C_{i'_3,i_2}D_{j_2,j'_1}
\iinner{i_1}{i'_2}\iinner{i_2}{i'_1}\iinner{i_3}{i'_3}
\iinner{j_1}{j'_2}\iinner{j_2}{j'_3}\iinner{j_3}{j'_1}\\
&&= \Tr{A}\Tr{DX}\Innerm{i_1}{BC}{i'_1}.
\end{eqnarray}
(12). If $(\pi,\sigma)=((12),(132))$, then
\begin{eqnarray}
&&\sum_{i_2,i_3,j_1,j_2,j_3,i'_2,i'_3,j'_1,j'_2,j'_3}
A_{j_1,j'_2}B_{i'_2,i_3}X_{j_3,j'_3}C_{i'_3,i_2}D_{j_2,j'_1}
\iinner{i_1}{i'_2}\iinner{i_2}{i'_1}\iinner{i_3}{i'_3}
\iinner{j_1}{j'_3}\iinner{j_2}{j'_1}\iinner{j_3}{j'_2}\\
&&= \Tr{D}\Tr{AX}\Innerm{i_1}{BC}{i'_1}.
\end{eqnarray}
(13). If $(\pi,\sigma)=((13),(1))$, then
\begin{eqnarray}
&&\sum_{i_2,i_3,j_1,j_2,j_3,i'_2,i'_3,j'_1,j'_2,j'_3}
A_{j_1,j'_2}B_{i'_2,i_3}X_{j_3,j'_3}C_{i'_3,i_2}D_{j_2,j'_1}
\iinner{i_1}{i'_3}\iinner{i_2}{i'_2}\iinner{i_3}{i'_1}
\iinner{j_1}{j'_1}\iinner{j_2}{j'_2}\iinner{j_3}{j'_3}\\
&&= \Tr{AD}\Tr{X}\Innerm{i_1}{CB}{i'_1}.
\end{eqnarray}
(14). If $(\pi,\sigma)=((13),(12))$, then
\begin{eqnarray}
&&\sum_{i_2,i_3,j_1,j_2,j_3,i'_2,i'_3,j'_1,j'_2,j'_3}
A_{j_1,j'_2}B_{i'_2,i_3}X_{j_3,j'_3}C_{i'_3,i_2}D_{j_2,j'_1}
\iinner{i_1}{i'_3}\iinner{i_2}{i'_2}\iinner{i_3}{i'_1}
\iinner{j_1}{j'_2}\iinner{j_2}{j'_1}\iinner{j_3}{j'_3}\\
&&= \Tr{A}\Tr{D}\Tr{X}\Innerm{i_1}{CB}{i'_1}.
\end{eqnarray}
(15). If $(\pi,\sigma)=((13),(13))$, then
\begin{eqnarray}
&&\sum_{i_2,i_3,j_1,j_2,j_3,i'_2,i'_3,j'_1,j'_2,j'_3}
A_{j_1,j'_2}B_{i'_2,i_3}X_{j_3,j'_3}C_{i'_3,i_2}D_{j_2,j'_1}
\iinner{i_1}{i'_3}\iinner{i_2}{i'_2}\iinner{i_3}{i'_1}
\iinner{j_1}{j'_3}\iinner{j_2}{j'_2}\iinner{j_3}{j'_1}\\
&&= \Tr{ADX}\Innerm{i_1}{CB}{i'_1}.
\end{eqnarray}
(16). If $(\pi,\sigma)=((13),(23))$, then
\begin{eqnarray}
&&\sum_{i_2,i_3,j_1,j_2,j_3,i'_2,i'_3,j'_1,j'_2,j'_3}
A_{j_1,j'_2}B_{i'_2,i_3}X_{j_3,j'_3}C_{i'_3,i_2}D_{j_2,j'_1}
\iinner{i_1}{i'_3}\iinner{i_2}{i'_2}\iinner{i_3}{i'_1}
\iinner{j_1}{j'_1}\iinner{j_2}{j'_3}\iinner{j_3}{j'_2}\\
&&= \Tr{DAX}\Innerm{i_1}{CB}{i'_1}.
\end{eqnarray}
(17). If $(\pi,\sigma)=((13),(123))$, then
\begin{eqnarray}
&&\sum_{i_2,i_3,j_1,j_2,j_3,i'_2,i'_3,j'_1,j'_2,j'_3}
A_{j_1,j'_2}B_{i'_2,i_3}X_{j_3,j'_3}C_{i'_3,i_2}D_{j_2,j'_1}
\iinner{i_1}{i'_3}\iinner{i_2}{i'_2}\iinner{i_3}{i'_1}
\iinner{j_1}{j'_2}\iinner{j_2}{j'_3}\iinner{j_3}{j'_1}\\
&&= \Tr{A}\Tr{DX}\Innerm{i_1}{CB}{i'_1}.
\end{eqnarray}
(18). If $(\pi,\sigma)=((13),(132))$, then
\begin{eqnarray}
&&\sum_{i_2,i_3,j_1,j_2,j_3,i'_2,i'_3,j'_1,j'_2,j'_3}
A_{j_1,j'_2}B_{i'_2,i_3}X_{j_3,j'_3}C_{i'_3,i_2}D_{j_2,j'_1}
\iinner{i_1}{i'_3}\iinner{i_2}{i'_2}\iinner{i_3}{i'_1}
\iinner{j_1}{j'_3}\iinner{j_2}{j'_1}\iinner{j_3}{j'_2}\\
&&= \Tr{D}\Tr{AX}\Innerm{i_1}{CB}{i'_1}.
\end{eqnarray}
(19). If $(\pi,\sigma)=((23),(1))$, then
\begin{eqnarray}
&&\sum_{i_2,i_3,j_1,j_2,j_3,i'_2,i'_3,j'_1,j'_2,j'_3}
A_{j_1,j'_2}B_{i'_2,i_3}X_{j_3,j'_3}C_{i'_3,i_2}D_{j_2,j'_1}
\iinner{i_1}{i'_1}\iinner{i_2}{i'_3}\iinner{i_3}{i'_2}
\iinner{j_1}{j'_1}\iinner{j_2}{j'_2}\iinner{j_3}{j'_3}\\
&&= \Tr{AD}\Tr{X}\Tr{B}\Tr{C}\iinner{i_1}{i'_{1}}.
\end{eqnarray}
(20). If $(\pi,\sigma)=((23),(12))$, then
\begin{eqnarray}
&&\sum_{i_2,i_3,j_1,j_2,j_3,i'_2,i'_3,j'_1,j'_2,j'_3}
A_{j_1,j'_2}B_{i'_2,i_3}X_{j_3,j'_3}C_{i'_3,i_2}D_{j_2,j'_1}
\iinner{i_1}{i'_1}\iinner{i_2}{i'_3}\iinner{i_3}{i'_2}
\iinner{j_1}{j'_2}\iinner{j_2}{j'_1}\iinner{j_3}{j'_3}\\
&&= \Tr{A}\Tr{D}\Tr{X}\Tr{B}\Tr{C}\iinner{i_1}{i'_{1}}.
\end{eqnarray}
(21). If $(\pi,\sigma)=((23),(13))$, then
\begin{eqnarray}
&&\sum_{i_2,i_3,j_1,j_2,j_3,i'_2,i'_3,j'_1,j'_2,j'_3}
A_{j_1,j'_2}B_{i'_2,i_3}X_{j_3,j'_3}C_{i'_3,i_2}D_{j_2,j'_1}
\iinner{i_1}{i'_1}\iinner{i_2}{i'_3}\iinner{i_3}{i'_2}
\iinner{j_1}{j'_3}\iinner{j_2}{j'_2}\iinner{j_3}{j'_1}\\
&&= \Tr{ADX}\Tr{B}\Tr{C}\iinner{i_1}{i'_1}.
\end{eqnarray}
(22). If $(\pi,\sigma)=((23),(23))$, then
\begin{eqnarray}
&&\sum_{i_2,i_3,j_1,j_2,j_3,i'_2,i'_3,j'_1,j'_2,j'_3}
A_{j_1,j'_2}B_{i'_2,i_3}X_{j_3,j'_3}C_{i'_3,i_2}D_{j_2,j'_1}
\iinner{i_1}{i'_1}\iinner{i_2}{i'_3}\iinner{i_3}{i'_2}
\iinner{j_1}{j'_1}\iinner{j_2}{j'_3}\iinner{j_3}{j'_2}\\
&&= \Tr{DAX}\Tr{B}\Tr{C}\iinner{i_1}{i'_1}.
\end{eqnarray}
(23). If $(\pi,\sigma)=((23),(123))$, then
\begin{eqnarray}
&&\sum_{i_2,i_3,j_1,j_2,j_3,i'_2,i'_3,j'_1,j'_2,j'_3}
A_{j_1,j'_2}B_{i'_2,i_3}X_{j_3,j'_3}C_{i'_3,i_2}D_{j_2,j'_1}
\iinner{i_1}{i'_1}\iinner{i_2}{i'_3}\iinner{i_3}{i'_2}
\iinner{j_1}{j'_2}\iinner{j_2}{j'_3}\iinner{j_3}{j'_1}\\
&&= \Tr{A}\Tr{DX}\Tr{B}\Tr{C}\iinner{i_1}{i'_1}.
\end{eqnarray}
(24). If $(\pi,\sigma)=((23),(132))$, then
\begin{eqnarray}
&&\sum_{i_2,i_3,j_1,j_2,j_3,i'_2,i'_3,j'_1,j'_2,j'_3}
A_{j_1,j'_2}B_{i'_2,i_3}X_{j_3,j'_3}C_{i'_3,i_2}D_{j_2,j'_1}
\iinner{i_1}{i'_1}\iinner{i_2}{i'_3}\iinner{i_3}{i'_2}
\iinner{j_1}{j'_3}\iinner{j_2}{j'_1}\iinner{j_3}{j'_2}\\
&&= \Tr{D}\Tr{AX}\Tr{B}\Tr{C}\iinner{i_1}{i'_1}.
\end{eqnarray}
(25). If $(\pi,\sigma)=((123),(1))$, then
\begin{eqnarray}
&&\sum_{i_2,i_3,j_1,j_2,j_3,i'_2,i'_3,j'_1,j'_2,j'_3}
A_{j_1,j'_2}B_{i'_2,i_3}X_{j_3,j'_3}C_{i'_3,i_2}D_{j_2,j'_1}
\iinner{i_1}{i'_2}\iinner{i_2}{i'_3}\iinner{i_3}{i'_1}
\iinner{j_1}{j'_1}\iinner{j_2}{j'_2}\iinner{j_3}{j'_3}\\
&&= \Tr{AD}\Tr{X}\Tr{C}\Innerm{i_1}{B}{i'_1}.
\end{eqnarray}
(26). If $(\pi,\sigma)=((123),(12))$, then
\begin{eqnarray}
&&\sum_{i_2,i_3,j_1,j_2,j_3,i'_2,i'_3,j'_1,j'_2,j'_3}
A_{j_1,j'_2}B_{i'_2,i_3}X_{j_3,j'_3}C_{i'_3,i_2}D_{j_2,j'_1}
\iinner{i_1}{i'_2}\iinner{i_2}{i'_3}\iinner{i_3}{i'_1}
\iinner{j_1}{j'_2}\iinner{j_2}{j'_1}\iinner{j_3}{j'_3}\\
&&= \Tr{A}\Tr{D}\Tr{X}\Tr{C}\Innerm{i_1}{B}{i'_1}.
\end{eqnarray}
(27). If $(\pi,\sigma)=((123),(13))$, then
\begin{eqnarray}
&&\sum_{i_2,i_3,j_1,j_2,j_3,i'_2,i'_3,j'_1,j'_2,j'_3}
A_{j_1,j'_2}B_{i'_2,i_3}X_{j_3,j'_3}C_{i'_3,i_2}D_{j_2,j'_1}
\iinner{i_1}{i'_2}\iinner{i_2}{i'_3}\iinner{i_3}{i'_1}
\iinner{j_1}{j'_3}\iinner{j_2}{j'_2}\iinner{j_3}{j'_1}\\
&&= \Tr{ADX}\Tr{C}\Innerm{i_1}{B}{i'_1}.
\end{eqnarray}
(28). If $(\pi,\sigma)=((123),(23))$, then
\begin{eqnarray}
&&\sum_{i_2,i_3,j_1,j_2,j_3,i'_2,i'_3,j'_1,j'_2,j'_3}
A_{j_1,j'_2}B_{i'_2,i_3}X_{j_3,j'_3}C_{i'_3,i_2}D_{j_2,j'_1}
\iinner{i_1}{i'_2}\iinner{i_2}{i'_3}\iinner{i_3}{i'_1}
\iinner{j_1}{j'_1}\iinner{j_2}{j'_3}\iinner{j_3}{j'_2}\\
&&= \Tr{DAX}\Tr{C}\Innerm{i_1}{B}{i'_1}.
\end{eqnarray}
(29). If $(\pi,\sigma)=((123),(123))$, then
\begin{eqnarray}
&&\sum_{i_2,i_3,j_1,j_2,j_3,i'_2,i'_3,j'_1,j'_2,j'_3}
A_{j_1,j'_2}B_{i'_2,i_3}X_{j_3,j'_3}C_{i'_3,i_2}D_{j_2,j'_1}
\iinner{i_1}{i'_2}\iinner{i_2}{i'_3}\iinner{i_3}{i'_1}
\iinner{j_1}{j'_2}\iinner{j_2}{j'_3}\iinner{j_3}{j'_1}\\
&&= \Tr{A}\Tr{DX}\Tr{C}\Innerm{i_1}{B}{i'_1}.
\end{eqnarray}
(30). If $(\pi,\sigma)=((123),(132))$, then
\begin{eqnarray}
&&\sum_{i_2,i_3,j_1,j_2,j_3,i'_2,i'_3,j'_1,j'_2,j'_3}
A_{j_1,j'_2}B_{i'_2,i_3}X_{j_3,j'_3}C_{i'_3,i_2}D_{j_2,j'_1}
\iinner{i_1}{i'_2}\iinner{i_2}{i'_3}\iinner{i_3}{i'_1}
\iinner{j_1}{j'_3}\iinner{j_2}{j'_1}\iinner{j_3}{j'_2}\\
&&= \Tr{D}\Tr{AX}\Tr{C}\Innerm{i_1}{B}{i'_1}.
\end{eqnarray}
(31). If $(\pi,\sigma)=((132),(1))$, then
\begin{eqnarray}
&&\sum_{i_2,i_3,j_1,j_2,j_3,i'_2,i'_3,j'_1,j'_2,j'_3}
A_{j_1,j'_2}B_{i'_2,i_3}X_{j_3,j'_3}C_{i'_3,i_2}D_{j_2,j'_1}
\iinner{i_1}{i'_3}\iinner{i_2}{i'_1}\iinner{i_3}{i'_2}
\iinner{j_1}{j'_1}\iinner{j_2}{j'_2}\iinner{j_3}{j'_3}\\
&&= \Tr{AD}\Tr{X}\Tr{B}\Innerm{i_1}{C}{i'_1}.
\end{eqnarray}
(32). If $(\pi,\sigma)=((132),(12))$, then
\begin{eqnarray}
&&\sum_{i_2,i_3,j_1,j_2,j_3,i'_2,i'_3,j'_1,j'_2,j'_3}
A_{j_1,j'_2}B_{i'_2,i_3}X_{j_3,j'_3}C_{i'_3,i_2}D_{j_2,j'_1}
\iinner{i_1}{i'_3}\iinner{i_2}{i'_1}\iinner{i_3}{i'_2}
\iinner{j_1}{j'_2}\iinner{j_2}{j'_1}\iinner{j_3}{j'_3}\\
&&= \Tr{A}\Tr{D}\Tr{X}\Tr{B}\Innerm{i_1}{C}{i'_1}.
\end{eqnarray}
(33). If $(\pi,\sigma)=((132),(13))$, then
\begin{eqnarray}
&&\sum_{i_2,i_3,j_1,j_2,j_3,i'_2,i'_3,j'_1,j'_2,j'_3}
A_{j_1,j'_2}B_{i'_2,i_3}X_{j_3,j'_3}C_{i'_3,i_2}D_{j_2,j'_1}
\iinner{i_1}{i'_3}\iinner{i_2}{i'_1}\iinner{i_3}{i'_2}
\iinner{j_1}{j'_3}\iinner{j_2}{j'_2}\iinner{j_3}{j'_1}\\
&&= \Tr{ADX}\Tr{B}\Innerm{i_1}{C}{i'_1}.
\end{eqnarray}
(34). If $(\pi,\sigma)=((132),(23))$, then
\begin{eqnarray}
&&\sum_{i_2,i_3,j_1,j_2,j_3,i'_2,i'_3,j'_1,j'_2,j'_3}
A_{j_1,j'_2}B_{i'_2,i_3}X_{j_3,j'_3}C_{i'_3,i_2}D_{j_2,j'_1}
\iinner{i_1}{i'_3}\iinner{i_2}{i'_1}\iinner{i_3}{i'_2}
\iinner{j_1}{j'_1}\iinner{j_2}{j'_3}\iinner{j_3}{j'_2}\\
&&= \Tr{DAX}\Tr{B}\Innerm{i_1}{C}{i'_1}.
\end{eqnarray}
(35). If $(\pi,\sigma)=((132),(123))$, then
\begin{eqnarray}
&&\sum_{i_2,i_3,j_1,j_2,j_3,i'_2,i'_3,j'_1,j'_2,j'_3}
A_{j_1,j'_2}B_{i'_2,i_3}X_{j_3,j'_3}C_{i'_3,i_2}D_{j_2,j'_1}
\iinner{i_1}{i'_3}\iinner{i_2}{i'_1}\iinner{i_3}{i'_2}
\iinner{j_1}{j'_2}\iinner{j_2}{j'_3}\iinner{j_3}{j'_1}\\
&&= \Tr{A}\Tr{DX}\Tr{B}\Innerm{i_1}{C}{i'_1}.
\end{eqnarray}
(36). If $(\pi,\sigma)=((132),(132))$, then
\begin{eqnarray}
&&\sum_{i_2,i_3,j_1,j_2,j_3,i'_2,i'_3,j'_1,j'_2,j'_3}
A_{j_1,j'_2}B_{i'_2,i_3}X_{j_3,j'_3}C_{i'_3,i_2}D_{j_2,j'_1}
\iinner{i_1}{i'_3}\iinner{i_2}{i'_1}\iinner{i_3}{i'_2}
\iinner{j_1}{j'_3}\iinner{j_2}{j'_1}\iinner{j_3}{j'_2}\\
&&= \Tr{D}\Tr{AX}\Tr{B}\Innerm{i_1}{C}{i'_1}.
\end{eqnarray}
Combing the 36 cases together gives the desired conclusion:
\begin{eqnarray}
\int UAU^\dagger BUXU^\dagger CUDU^\dagger \dif\mu(U)=
\mu_1\cdot\I_d+\mu_2\cdot BC+\mu_3\cdot CB + \mu_4\cdot B +
\mu_5\cdot C,
\end{eqnarray}
where the coefficients $\mu_j(j=1,\ldots,5)$ are given below:
\begin{eqnarray}
\mu_1&:=& \mathrm{Wg}(1,1,1)\Tr{AD}\Tr{X}\Tr{BC} +
\mathrm{Wg}(2,1)\Tr{A}\Tr{D}\Tr{X}\Tr{BC}\notag\\
&&+\mathrm{Wg}(2,1)\Tr{ADX}\Tr{BC}+\mathrm{Wg}(2,1)\Tr{DAX}\Tr{BC}\notag\\
&&+\mathrm{Wg}(3)\Tr{A}\Tr{DX}\Tr{BC}+\mathrm{Wg}(3)\Tr{D}\Tr{AX}\Tr{BC}\notag\\
&&+\mathrm{Wg}(2,1)\Tr{AD}\Tr{X}\Tr{B}\Tr{C}+\mathrm{Wg}(3)\Tr{A}\Tr{D}\Tr{X}\Tr{B}\Tr{C}\notag\\
&&+\mathrm{Wg}(3)\Tr{ADX}\Tr{B}\Tr{C}+\mathrm{Wg}(1,1,1)\Tr{DAX}\Tr{B}\Tr{C}\notag\\
&&+\mathrm{Wg}(2,1)\Tr{A}\Tr{DX}\Tr{B}\Tr{C}+\mathrm{Wg}(2,1)\Tr{D}\Tr{AX}\Tr{B}\Tr{C},\\
\mu_2&:=& \mathrm{Wg}(2,1)\Tr{AD}\Tr{X} +
\mathrm{Wg}(1,1,1)\Tr{A}\Tr{D}\Tr{X}+\mathrm{Wg}(3)\Tr{ADX}\notag\\
&&+\mathrm{Wg}(3)\Tr{DAX}+\mathrm{Wg}(2,1)\Tr{A}\Tr{DX}+\mathrm{Wg}(2,1)\Tr{D}\Tr{AX},\\
\mu_3&:=& \mathrm{Wg}(2,1)\Tr{AD}\Tr{X} +
\mathrm{Wg}(3)\Tr{A}\Tr{D}\Tr{X}+\mathrm{Wg}(1,1,1)\Tr{ADX}\notag\\
&&+\mathrm{Wg}(3)\Tr{DAX}+\mathrm{Wg}(2,1)\Tr{A}\Tr{DX}+\mathrm{Wg}(2,1)\Tr{D}\Tr{AX},\\
\mu_4&:=& \mathrm{Wg}(3)\Tr{AD}\Tr{X}\Tr{C} +
\mathrm{Wg}(2,1)\Tr{A}\Tr{D}\Tr{X}\Tr{C} \notag\\
&&+ \mathrm{Wg}(2,1)\Tr{ADX}\Tr{C}+\mathrm{Wg}(2,1)\Tr{DAX}\Tr{C}\notag\\
&&+\mathrm{Wg}(1,1,1)\Tr{A}\Tr{DX}\Tr{C}+\mathrm{Wg}(3)\Tr{D}\Tr{AX}\Tr{C},\\
\mu_5&:=& \mathrm{Wg}(3)\Tr{AD}\Tr{X}\Tr{B} +
\mathrm{Wg}(2,1)\Tr{A}\Tr{D}\Tr{X}\Tr{B} \notag\\
&&+ \mathrm{Wg}(2,1)\Tr{ADX}\Tr{B}+\mathrm{Wg}(2,1)\Tr{DAX}\Tr{B}\notag\\
&&+\mathrm{Wg}(3)\Tr{A}\Tr{DX}\Tr{B}+\mathrm{Wg}(1,1,1)\Tr{D}\Tr{AX}\Tr{B}.
\end{eqnarray}
We are done.
\end{proof}

\begin{remark}
$\mathrm{Wg}:=\frac1{(k!)^2}\sum_{\lambda\vdash
k}\frac{\dim(\bP_\lambda)^2}{\dim(\bQ_\lambda)}\chi_\lambda$ is
called the \emph{Weingarten function} \cite{Benoit}. In particular,
for $\lambda\vdash 3$, we have:
\begin{eqnarray}
\mathrm{Wg}(1,1,1) &=& \frac{d^2-2}{d(d^2-1)(d^2-4)}=\Pa{d-\frac2d}\cdot\frac1{N_d},\\
\mathrm{Wg}(2,1) &=& -\frac1{(d^2-1)(d^2-4)}=(-1)\cdot\frac1{N_d},\\
\mathrm{Wg}(3) &=& \frac{2}{d(d^2-1)(d^2-4)} =
\frac2d\cdot\frac1{N_d},
\end{eqnarray}
where $N_d=(d^2-1)(d^2-4)$. With these coefficients, we then have
\begin{eqnarray}
N_d\mu_1&:=& \Pa{d-\frac2d}\Tr{AD}\Tr{X}\Tr{BC} +\Pa{d-\frac2d}\Tr{DAX}\Tr{B}\Tr{C}+\frac2d\Tr{A}\Tr{DX}\Tr{BC}\notag\\
&&+\frac2d\Tr{D}\Tr{AX}\Tr{BC}+\frac2d\Tr{A}\Tr{D}\Tr{X}\Tr{B}\Tr{C}+\frac2d\Tr{ADX}\Tr{B}\Tr{C}\notag\\
&&- \Tr{A}\Tr{D}\Tr{X}\Tr{BC}-\Tr{ADX}\Tr{BC}-\Tr{DAX}\Tr{BC}\notag\\
&&-\Tr{AD}\Tr{X}\Tr{B}\Tr{C}-\Tr{A}\Tr{DX}\Tr{B}\Tr{C}-\Tr{D}\Tr{AX}\Tr{B}\Tr{C},\\
N_d\mu_2&:=&
\Pa{d-\frac2d}\Tr{A}\Tr{D}\Tr{X}+\frac2d\Tr{ADX}+\frac2d\Tr{DAX}\notag\\
&&-\Tr{AD}\Tr{X} -\Tr{A}\Tr{DX}-\Tr{D}\Tr{AX},\\
N_d\mu_3&:=& \frac2d\Tr{DAX} +
\frac2d\Tr{A}\Tr{D}\Tr{X}+\Pa{d-\frac2d}\Tr{ADX}\notag\\
&&-\Tr{AD}\Tr{X}-\Tr{A}\Tr{DX}-\Tr{D}\Tr{AX},\\
N_d\mu_4&:=& \Pa{d-\frac2d}\Tr{A}\Tr{DX}\Tr{C}+\frac2d\Tr{AD}\Tr{X}\Tr{C} + \frac2d\Tr{D}\Tr{AX}\Tr{C}\notag\\
&&- \Tr{A}\Tr{D}\Tr{X}\Tr{C} - \Tr{ADX}\Tr{C} - \Tr{DAX}\Tr{C},\\
N_d\mu_5&:=& \frac2d\Tr{AD}\Tr{X}\Tr{B} +\frac2d\Tr{A}\Tr{DX}\Tr{B}+\Pa{d-\frac2d}\Tr{D}\Tr{AX}\Tr{B} \notag\\
&&- \Tr{A}\Tr{D}\Tr{X}\Tr{B}- \Tr{ADX}\Tr{B} - \Tr{DAX}\Tr{B}.
\end{eqnarray}
\end{remark}

\subsection*{B. The proof of Theorem~\ref{th:A1}}

By Lemma~\ref{lem:sixth-moment},
\begin{eqnarray*}
\int U^\dagger M_{ij}UTU^\dagger M_{jl}UTU^\dagger M_{li}U\dif\mu(U)
= f(i,j,l)\cdot\I_d + g(i,j,l)\cdot T + h(i,j,l)\cdot T^2,
\end{eqnarray*}
where
\begin{eqnarray*}
f(i,j,l) &=& \mathrm{Wg}(1,1,1)d^2_A\delta_{jl}\Tr{T^2} +
\mathrm{Wg}(2,1)d^3_A\delta_{ij}\delta_{jl}\delta_{li}\Tr{T^2}+\mathrm{Wg}(2,1)d_A\delta_{ij}\delta_{jl}\delta_{li}\Tr{T^2}\\
&&+\mathrm{Wg}(2,1)d_A\Tr{T^2}+\mathrm{Wg}(3)d^2_A\delta_{ij}\Tr{T^2}+\mathrm{Wg}(3)d^2_A\delta_{il}\Tr{T^2}\\
&&+\mathrm{Wg}(2,1)d^2_A\delta_{jl}\Tr{T}^2+\mathrm{Wg}(3)d^3_A\delta_{ij}\delta_{jl}\delta_{li}\Tr{T}^2+\mathrm{Wg}(3)d_A\delta_{ij}\delta_{jl}\delta_{li}\Tr{T}^2\\
&&+\mathrm{Wg}(1,1,1)d_A\Tr{T}^2+\mathrm{Wg}(2,1)d^2_A\delta_{ij}\Tr{T}^2+\mathrm{Wg}(2,1)d^2_A\delta_{il}\Tr{T}^2,\\
g(i,j,l) &=& \mathrm{Wg}(3)d^2_A\delta_{jl}\Tr{T}  +
\mathrm{Wg}(2,1)d^3_A\delta_{ij}\delta_{jl}\delta_{li}\Tr{T} + \mathrm{Wg}(2,1)d_A\delta_{ij}\delta_{jl}\delta_{li}\Tr{T}\\
&&+\mathrm{Wg}(2,1)d_A\Tr{T}+\mathrm{Wg}(1,1,1)d^2_A\delta_{ij}\Tr{T}+\mathrm{Wg}(3)d^2_A\delta_{il}\Tr{T}\\
&&+\mathrm{Wg}(3)d^2_A\delta_{jl}\Tr{T} +
\mathrm{Wg}(2,1)d^3_A\delta_{ij}\delta_{jl}\delta_{li}\Tr{T} + \mathrm{Wg}(2,1)d_A\delta_{ij}\delta_{jl}\delta_{li}\Tr{T}\\
&&+\mathrm{Wg}(2,1)d_A\Tr{T}+\mathrm{Wg}(3)d^2_A\delta_{ij}\Tr{T}+\mathrm{Wg}(1,1,1)d^2_A\delta_{il}\Tr{T},\\
h(i,j,l) &=& \mathrm{Wg}(2,1)d^2_A\delta_{jl} +
\mathrm{Wg}(1,1,1)d^3_A\delta_{ij}\delta_{jl}\delta_{li}+\mathrm{Wg}(3)d_A\delta_{ij}\delta_{jl}\delta_{li}\\
&&+\mathrm{Wg}(3)d_A+\mathrm{Wg}(2,1)d^2_A\delta_{ij}+\mathrm{Wg}(2,1)d^2_A\delta_{il},\notag\\
&& +\mathrm{Wg}(2,1)d^2_A\delta_{jl} +
\mathrm{Wg}(3)d^3_A\delta_{ij}\delta_{jl}\delta_{li}+\mathrm{Wg}(1,1,1)d_A\delta_{ij}\delta_{jl}\delta_{li}\\
&&+\mathrm{Wg}(3)d_A+\mathrm{Wg}(2,1)d^2_A\delta_{ij}+\mathrm{Wg}(2,1)d^2_A\delta_{il}.
\end{eqnarray*}
Note that the meaning of the notation $\mathrm{Wg}(*)$ can be found
in the Appendix. Thus for $f:=\sum^{d_B}_{i,j,l=1}f(i,j,l)$,
$g:=\sum^{d_B}_{i,j,l=1}g(i,j,l)$, and
$h:=\sum^{d_B}_{i,j,l=1}h(i,j,l)$, we have
\begin{eqnarray*}
f &=& \Pa{[\mathrm{Wg}(1,1,1)+2\mathrm{Wg}(3)]d^2+
\mathrm{Wg}(2,1)(dd^2_A+d+dd^2_B)}\Tr{T^2}\\
&&+\Pa{3\mathrm{Wg}(2,1)d^2+\mathrm{Wg}(3)(dd^2_A+d)+\mathrm{Wg}(1,1,1)dd^2_B}\Tr{T}^2,\\
g &=& \Pa{[4\mathrm{Wg}(3) + 2\mathrm{Wg}(1,1,1)]d^2+
2\mathrm{Wg}(2,1)(dd^2_A + d +dd^2_B)}\Tr{T},\\
h &=& 6\mathrm{Wg}(2,1)d^2 +
[\mathrm{Wg}(1,1,1)+\mathrm{Wg}(3)]dd^2_A+[\mathrm{Wg}(1,1,1)+\mathrm{Wg}(3)]d+2\mathrm{Wg}(3)dd^2_B.
\end{eqnarray*}
Hence, for $T=\I_A\ot\I_B/d_B - \rho_{AB}$, $\Tr{T}=d_A-1$ and
$\Tr{T^2}=\frac{d_A-2}{d_B}+\Tr{\rho^2_{AB}}$, then
\begin{eqnarray}
f&=& \frac{d(d^2-d^2_A-d^2_B+1)}{(d^2-1)(d^2-4)}\Tr{T^2} + \frac{d^2(d^2_B-3)+2(d^2_A-d^2_B+1)}{(d^2-1)(d^2-4)}\Tr{T}^2,\\
&=&\frac{d_A(d^2_A-1)(d^2_B-1)}{(d^2-1)(d^2-4)}\Pa{d_A+d_B\Tr{\rho^2_{AB}}-2}\\
&&+\frac{(d^2-2d^2_A-2)(d^2_B-1)}{(d^2-1)(d^2-4)}(d_A-1)^2,\\
g&=&\frac{2d(d^2-d^2_A-d^2_B+1)}{(d^2-1)(d^2-4)}\Tr{T} = \frac{2d(d_A-1)(d^2_A-1)(d^2_B-1)}{(d^2-1)(d^2-4)},\\
h&=&\frac{d^2(d^2_A-5)+4d^2_B}{(d^2-1)(d^2-4)} =
\frac{(d^2_A-1)(d^2_A-4)d^2_B}{(d^2-1)(d^2-4)}.
\end{eqnarray}
Therefore,
\begin{eqnarray}
\int U^\dagger\Br{\Gamma(UTU^\dagger)}^2U\dif\mu(U) =f\cdot\I_d +
g\cdot T + h\cdot T^2,
\end{eqnarray}
implying that
\begin{eqnarray}
a_2 &=& f + g\cdot \Tr{\rho_{AB}T} + h\cdot \Tr{\rho_{AB}T^2}\notag\\
&=&f + g\cdot\Pa{\frac1{d_B}-\Tr{\rho^2_{AB}}} + h\cdot
\Pa{\frac1{d^2_B} - \frac2{d_B}\Tr{\rho^2_{AB}}
+\Tr{\rho^3_{AB}}}\notag\\
&=& \Pa{f + g\frac1{d_B} + h\frac1{d^2_B}} -
\Pa{g+h\frac2{d_B}}\Tr{\rho^2_{AB}} + h\Tr{\rho^3_{AB}}.
\end{eqnarray}
We make further analysis of the term $a_n$ although we have already
known the fact that $\lim_{n\to \infty}a_n=0$. Let
$\varphi_\rho(X)=\Tr{\rho X}$. Apparently, $\varphi_\rho$ is a
positive unital linear mapping (in fact, it is a positive unital
linear functional from the set of $n\times n$ Hermitian matrices to
$\real$). It is easily seen that $f(x)=x^n$ is a convex function
from $\real$ to $\real$. By using \cite[Theorem~4.15,
pp147]{Hiai2014}, we see that
\begin{eqnarray*}
a_n &=& \varphi_\rho\Pa{\int
\Br{U^\dagger\Gamma(UTU^\dagger)U}^n\dif\mu(U)} =
\varphi_\rho\Pa{\int f(U^\dagger\Gamma(UTU^\dagger)U)\dif\mu(U)}\\
&=& \int \dif\mu(U)(\varphi_\rho\circ
f)(U^\dagger\Gamma(UTU^\dagger)U)\geqslant \int \dif\mu(U)
f[\varphi_\rho(U^\dagger\Gamma(UTU^\dagger)U)]\\
&\geqslant&\int \dif\mu(U) \Br{\Tr{\rho
U^\dagger\Gamma(UTU^\dagger)U}}^n\geqslant
\Pa{\Tr{\rho\int\dif\mu(U)U^\dagger\Gamma(UTU^\dagger)U}}^n.
\end{eqnarray*}
By \eqref{eq:1st-matrix-int}, we have
\begin{eqnarray*}
\Tr{\rho\int \dif\mu(U)U^\dagger\Gamma(UTU^\dagger)U} =
\frac{d_A-1}{d^2-1}\Br{(1+dd_B) - (d+d_B)\Tr{\rho^2_{AB}}}=a_1.
\end{eqnarray*}
Thus
\begin{eqnarray}
a_n &\geqslant& a^n_1.
\end{eqnarray}
Then
$$
a_1=\frac{(d_A-1)(d_B-1)}{d+1}+\frac{(d_A-1)(d+d_B)}{d^2-1}\rS_L(\rho_{AB}),
$$
where $\rS_L(\rho_{AB})\in\Br{0,1-\frac1d}$. Clearly
\begin{eqnarray}
\frac{(d_A-1)(d_B-1)}{d+1}\leqslant a_1\leqslant
\frac{(d_A-1)\Pa{1+\frac1{d_A}}}{d+1}<1.
\end{eqnarray}
Now
\begin{eqnarray}
\int\rS(\rho'_A)\dif\mu(U)=\sum^\infty_{n=1}\frac{a_n}n\geqslant
\sum^\infty_{n=1}\frac{a^n_1}n = -\ln(1-a_1).
\end{eqnarray}
We see from the above lower bound, i.e., $-\ln(1-a_1)$, that when
the purity of $\rho_{AB}$ decreases, $a_1$ increases. Hence such
lower bound will be tighter.

\subsection*{C. The proof of Proposition~\ref{prop:ave-corr}}

Clearly, the first inequality is easily obtained
\begin{eqnarray}
\int I(A:B)_{\rho'}\dif\mu(U) &\geqslant& \rS\Pa{\rho_{AB}||\int \ln
\Br{U^\dagger
(\rho'_A\ot\rho'_B)U}\dif\mu(U)}\\
&=& \rS\Pa{\rho_{AB}||c_0\cdot\I_d + c_1\cdot\rho_{AB} +
c_2\cdot\rho^2_{AB}}.
\end{eqnarray}
In order to show the second inequality, note that, for any two
density matrices $\rho$ and $\sigma$,
\begin{eqnarray}
\rF(\rho,\sigma)\geqslant \Tr{\sqrt{\rho}\sqrt{\sigma}}\geqslant
\Tr{\rho\sigma}.
\end{eqnarray}
Then, for $\rho'_{AB}=U\rho_{AB}U^\dagger$, we have
\begin{eqnarray}
\rF(\rho'_{AB},\rho'_A\ot\rho'_B)\geqslant
\Tr{\rho'_{AB}\rho'_A\ot\rho'_B} =
\Tr{\rho_{AB}U^\dagger\Pa{\rho'_A\ot\rho'_B} U},
\end{eqnarray}
implying that
\begin{eqnarray}
\int\rF(\rho'_{AB},\rho'_A\ot\rho'_B)\dif\mu(U)\geqslant
\int\Tr{\rho_{AB}U^\dagger\Pa{\rho'_A\ot\rho'_B} U}\dif\mu(U) =
c_0+c_1\Tr{\rho^2_{AB}}+c_2\Tr{\rho^3_{AB}}.
\end{eqnarray}
By the concavity of fidelity, we have
\begin{eqnarray}
\int\rF(\rho'_{AB},\rho'_A\ot\rho'_B)\dif\mu(U) =
\int\rF(\rho_{AB},U^\dagger(\rho'_A\ot\rho'_B)U)\dif\mu(U)\leqslant
\rF(\rho_{AB},c_0\cdot\I_d + c_1\cdot\rho_{AB} +
c_2\cdot\rho^2_{AB}).
\end{eqnarray}
Therefore
\begin{eqnarray}
c_0+c_1\Tr{\rho^2_{AB}}+c_2\Tr{\rho^3_{AB}}\leqslant\int\rF(\rho'_{AB},\rho'_A\ot\rho'_B)\dif\mu(U)\leqslant
\rF(\rho_{AB},c_0\cdot\I_d + c_1\cdot\rho_{AB} +
c_2\cdot\rho^2_{AB}).
\end{eqnarray}
This completes the proof.

\subsection*{D. The proof of Theorem~\ref{th:A2}}

Note that $\rho'_{AB}=U\rho_{AB}U^\dagger$. We see from
\eqref{eq:averageQMI-lower} that
\begin{eqnarray}
\int I(A:B)_{\rho'}\dif\mu(U)\geqslant\rS\Pa{\rho_{AB}
||c_0\cdot\I_d + c_1\cdot\rho_{AB} + c_2\cdot\rho^2_{AB}}.
\end{eqnarray}
Since $I(A:B)_{\rho'} = \rS(\rho'_A)+\rS(\rho'_B)-\rS(\rho_{AB})$,
it follows that
\begin{eqnarray}
\int
\Pa{\rS(\rho'_A)+\rS(\rho'_B)-\rS(\rho_{AB})}\dif\mu(U)\geqslant\rS\Pa{\rho_{AB}
||c_0\cdot\I_d + c_1\cdot\rho_{AB} + c_2\cdot\rho^2_{AB}}.
\end{eqnarray}
That is,
\begin{eqnarray}
\langle S_A+S_B\rangle\geqslant \rS(\rho_{AB})+\rS\Pa{\rho_{AB}
||c_0\cdot\I_d + c_1\cdot\rho_{AB} + c_2\cdot\rho^2_{AB}}.
\end{eqnarray}
This confirms the first inequality. Besides, by
Eq.~\eqref{eq:new-id}, we get
\begin{eqnarray}
&&\rS\Pa{c_0\cdot\I_d + c_1\cdot\rho_{AB} + c_2\cdot\rho^2_{AB}} =
\rS\Pa{\int U^\dagger(\rho'_A\ot \rho'_B)U \dif\mu(U)}\\
&&\geqslant\int \rS\Pa{U^\dagger(\rho'_A\ot
\rho'_B)U}\dif\mu(U)=\int\rS(\rho'_A)\dif\mu(U)+\int\rS(\rho'_B)\dif\mu(U).
\end{eqnarray}
This confirms the second inequality. Therefore we have
\begin{eqnarray*}
\rS(\rho_{AB})+\rS\Pa{\rho_{AB} ||c_0\cdot\I_d + c_1\cdot\rho_{AB} +
c_2\cdot\rho^2_{AB}}\leqslant \rS\Pa{c_0\cdot\I_d +
c_1\cdot\rho_{AB} + c_2\cdot\rho^2_{AB}}.
\end{eqnarray*}
This is equivalent to the following
\begin{eqnarray*}
\rS\Pa{\rho_{AB} ||c_0\cdot\I_d + c_1\cdot\rho_{AB} +
c_2\cdot\rho^2_{AB}}\leqslant \rS\Pa{c_0\cdot\I_d +
c_1\cdot\rho_{AB} + c_2\cdot\rho^2_{AB}} - \rS(\rho_{AB}).
\end{eqnarray*}
Next we show that $\rS\Pa{c_0\cdot\I_d + c_1\cdot\rho_{AB} +
c_2\cdot\rho^2_{AB}} = \rS(\rho_{AB})$ if and only if $\rho_{AB}$ is
maximally mixed state. Clearly if $\rho_{AB}$ is maximally mixed
state, i.e., $\rS(\rho_{AB})=\ln(d)$, since $\rS\Pa{c_0\cdot\I_d +
c_1\cdot\rho_{AB} + c_2\cdot\rho^2_{AB}} -
\rS(\rho_{AB})\geqslant0$, then $\rS\Pa{c_0\cdot\I_d +
c_1\cdot\rho_{AB} + c_2\cdot\rho^2_{AB}}\geqslant \ln(d)$,
apparently $\rS\Pa{c_0\cdot\I_d + c_1\cdot\rho_{AB} +
c_2\cdot\rho^2_{AB}}\leqslant \ln(d)$, thus
$\rS(\rho_{AB})=\rS\Pa{c_0\cdot\I_d + c_1\cdot\rho_{AB} +
c_2\cdot\rho^2_{AB}}=\ln(d)$, the maximum of von Neuman entropy.
Reversely, if $\rS\Pa{c_0\cdot\I_d + c_1\cdot\rho_{AB} +
c_2\cdot\rho^2_{AB}} = \rS(\rho_{AB})$, then by the obtained
inequality, we have
$$
\rS\Pa{\rho_{AB} ||c_0\cdot\I_d + c_1\cdot\rho_{AB} +
c_2\cdot\rho^2_{AB}}=0,
$$
which holds if and only if $\rho_{AB} = c_0\cdot\I_d +
c_1\cdot\rho_{AB} + c_2\cdot\rho^2_{AB}$. This means that for any
eigenvalue $\lambda(\geqslant0)$ of $\rho_{AB}$ must satisfy that
$$
c_2\lambda^2 +(c_1-1)\lambda+c_0=0.
$$
Solve this equation, we get
$$
\lambda = \frac{(1-c_1)-\sqrt{(1-c_1)^2-4c_0c_2}}{2c_2}=\frac1{d}.
$$
Note that we have dropped another root being larger than one. Thus
$\rho_{AB}$ is maximally mixed state. In fact, we get that
$\rho_{AB} = c_0\cdot\I_d + c_1\cdot\rho_{AB} + c_2\cdot\rho^2_{AB}$
if and only if $\rho_{AB}$ is maximally mixed state.



\begin{thebibliography}{99}

\bibitem{Bravyi2004}
S. Bravyi, {\em Compatibility between local and multipartite
states}, Quant Inf. \& Comput. \textbf{4}(1), 012-026 (2004).

\bibitem{Choi1975}
M-D. Choi, {\em Completely positive linear maps on complex
matrices}, Linear Alg Appl.
\href{http://dx.doi.org/10.1016/0024-3795(75)90075-0}{\textbf{10},
285-290 (1975).}

\bibitem{Ch2014}
M. Christandl, B. Doran, S. Kousidis, M. Walter, {\em Eigenvalue
distributions of reduced density matrices}, \cmp
\href{http://doi.dx.org/10.1007/s00220-014-2144-4}{\textbf{332},
1-52 (2014).}

\bibitem{Benoit}
B. Collins, {\em Moments and Cumulants of Polynomial Random
Variables on Unitary Groups, the Itzykson-Zuber Integral, and Free
Probability}, Int. Math. Res. Not. \textbf{17}, 953 (2003).

\bibitem{Collins}
B. Collins, P. \'{S}niady, {\em Integration with Respect to the Haar
Measure on Unitary, Orthogonal and Symplectic Group}, \cmp
\href{http://doi.dx.org/10.1007/s00220-006-1554-3}{\textbf{264}(3),
773-795 (2006).}

\bibitem{Dyer2014}
J.P. Dyer, {\em Divergence of Lubkin's series for a quantum
subsystem's mean entropy},
\href{http://arxiv.org/abs/1406.5776}{arXiv:1406.5776}

\bibitem{FK1994}
S.K. Foong and S. Kano, {\em Proof of Page's conjecture on the
avearge entropy of a subsystem}, \prl
\href{http://dx.doi.org/10.1103/PhysRevLett.72.1148}{\textbf{72},
1148 (1994).}

\bibitem{MH2013}
M. Gessner and H-P Breuer, {\em Generic features of the dynamics of
complex open quantum systems: Stataistical approach based on
averages over the unitary group}, \pre
\href{http://dx.doi.org/10.1103/PhysRevE.87.042128}{\textbf{87},
042128 (2013).}

\bibitem{Giorda2017}
P. Giorda and M. Allegra, {\em Two-qubit correlations revisited:
average mutual information, relevant (and useful) observables and an
application to remote state preparation},
\href{https://arxiv.org/abs/1606.02197}{arXiv: 1606.02197}

\bibitem{Hiai2014}
F. Hiai and D. Petz, {\em Introduction to Matrix Analysis and
Applications}, Hindustan Book Agency, Springer (2014).

\bibitem{Jevtic2012a}
S. Jevtic, D. Jennings, and T. Rudolph, {\em Maximally and minimally
correlated states attainable within a closed evolving system}, \prl
\href{http://dx.doi.org/10.1103/PhysRevLett.108.110403}{\textbf{108},
110403 (2012).}

\bibitem{Jevtic2012b}
S. Jevtic, D. Jennings, and T. Rudolph, {\em Quantum mutual
information along unitary orbits}, \pra
\href{http://dx.doi.org/10.1103/PhysRevA.85.052121}{\textbf{85},
052121 (2012).}

\bibitem{Klyachko2006}
A. Klyachko, {\em Quantum marginal problem and
$N$-representability}, J. Phys. Conference Series \textbf{36}, 72-86
(2006).


\bibitem{AL2006}
A. Lachal, {\em Probabilistic approach to Page's formula for the
entropy of a quantum system}, Stochasitcs: An International Journal
of Probability and Stochastics Processes \textbf{78}, 157-178
(2006).

\bibitem{Lub1978}
E. Lubkin, {\em Entropy of an $n$-system from its correlation with a
$k$-reservoir}, \jmp
\href{http://dx.doi.org/10.1063/1.523763}{\textbf{19}(5), 1028
(1978).}

\bibitem{Modi}
K. Modi, M. Gu, {\em Coherent and incoherent contents of
correlations}, Int. J. Mod. Phys. B
\href{http://dx.doi.org/10.1142/S0217979213450276}{\textbf{27},
1345027 (2012).}

\bibitem{OK2014}
M. Oszmaniec, M. Ku\'{s}, {\em Fraction of isospectral states
exhibiting quantum correlations}, \pra
\href{http://dx.doi.org/10.1103/PhysRevA.90.010302}{\textbf{90},
010302 (2014).}


\bibitem{Page1993}
D.N. Page, {\em Average entropy of a subsystem}, \prl
\href{http://dx.doi.org/10.1103/PhysRevLett.71.1291}{\textbf{71},
1291 (1993).}

\bibitem{Sen1996}
S. Sen, {\em Average entropy of a quantum subsystem}, \prl
\href{http://dx.doi.org/10.1103/PhysRevLett.77.1}{\textbf{77}, 1
(1996).}

\bibitem{SR1995}
J. S\'{a}nchez-Ruiz, {\em Simple proof of Page's conjecture on the
average entropy of a subsystem}, \pre
\href{http://dx.doi.org/10.1103/PhysRevE.52.5653}{\textbf{52}, 5653
(1995).}

\bibitem{Walter2013}
M. Walter, B. Doran, D. Gross, and M. Christandl, {\em Entanglement
Polytopes: Multiparticle Entanglement from Single-Particle
Information}, \sci
\href{http://dx.doi.org/10.1126/science.1232957}{\textbf{340}, 6137
(2013).}

\bibitem{LZ2014}
L. Zhang, {\em Matrix integrals over unitary groups: An application
of Schur-Weyl duality},
\href{http://arxiv.org/abs/1408.3782}{arXiv:1408.3782v3}


\bibitem{lin2017}
L. Zhang, {\em Average coherence and its typicality for random mixed
quantum states}, \jpa: Math. Theor.
\href{https://doi.org/10.1088/1751-8121/aa6179}{to appear}

\bibitem{ZF2014}
L. Zhang, S-M. Fei, {\em Quantum fidelity and relative entropy
between unitary orbits}, \jpa: Math. Theor.
\href{http://dx.doi.org/10.1088/1751-8113/47/5/055301}{\textbf{47},
055301 (2014).}

\bibitem{ZCB2015}
L. Zhang, L. Chen, and K. Bu, {\em Fidelity between one bipartite
quantum state and another undergoing local unitary dynamics}, Quant.
Inf. Process
\href{http://dx.doi.org/10.1007/s11128-015-1117-7}{\textbf{14}, 4715
(2015).}


\bibitem{LUA2017}
L. Zhang, U. Singh, and A.K. Pati, {\em Average subentropy,
coherence and entanglement of random mixed quantum states}, \aop
\href{http://dx.doi.org/10.1016/j.aop.2016.12.024}{\textbf{377},
125-146 (2017).}


\end{thebibliography}
\end{document}